\documentclass[a4paper,english,cleveref,thm-restate]{lipics-v2021}

\pdfoutput=1
\hideLIPIcs
\nolinenumbers

\usepackage{mathtools}
\usepackage{wrapfig}
\usepackage{xspace}

\graphicspath{{./figures/}}

\newcommand{\etal}{\textit{et~al.\@}\xspace}

\newcommand{\f}{Fr\'echet\xspace}
\newcommand{\eps}{\ensuremath{\varepsilon}}
\newcommand{\from}{\ensuremath{\colon}}

\newcommand{\R}{\ensuremath{\mathbb{R}}}
\newcommand{\Z}{\ensuremath{\mathbb{Z}}}
\newcommand{\G}{\ensuremath{\mathbb{G}}}

\newcommand{\D}{\ensuremath{\mathcal{D}}}
\newcommand{\F}[1][\delta]{\ensuremath{\mathcal{F}_{\leq #1}}}

\title{Faster Fr\'echet Distance Approximation through Truncated Smoothing\footnote{%
This paper is the full version of the papers ``A Subquadratic $n^\varepsilon$-approximation for the Continuous Fréchet Distance''~\cite{vanderhorst23continuous_frechet}, which appeared at SODA 2023, and ``Faster Fr\'echet Distance Approximation through Truncated Smoothing''~\cite{vanderHorst24faster_Frechet}, which appeared at SoCG 2024.
}}

\author
{Thijs van der Horst}
{Department of Information and Computing Sciences, Utrecht University, the Netherlands
\and
Department of Mathematics and Computer Science, TU Eindhoven, the Netherlands}
{t.w.j.vanderhorst@uu.nl}
{https://orcid.org/0009-0002-6987-4489}
{}

\author
{Marc van Kreveld}
{Department of Information and Computing Sciences, Utrecht University, The Netherlands}
{m.j.vankreveld@uu.nl}
{https://orcid.org/0000-0001-8208-3468}
{}

\author
{Tim Ophelders}
{Department of Information and Computing Sciences, Utrecht University, the Netherlands
\and
Department of Mathematics and Computer Science, TU Eindhoven, the Netherlands}
{t.a.e.ophelders@uu.nl}
{https://orcid.org/0000-0002-9570-024X}
{partially supported by the Dutch Research Council (NWO) under project no.\ VI.Veni.212.260.}

\author
{Bettina Speckmann}
{Department of Mathematics and Computer Science, TU Eindhoven, The Netherlands}
{b.speckmann@tue.nl}
{https://orcid.org/0000-0002-8514-7858}
{}

\authorrunning{T. van der Horst, M. van Kreveld, and T. Ophelders, and B. Speckmann}
\Copyright{Thijs van der Horst, Marc van Kreveld, Tim Ophelders, and Bettina Speckmann}



\category{}

\relatedversion{}

\begin{document}

\maketitle

\begin{abstract}
    The Fr\'echet distance is a commonly used distance measure for curves.
    Computing the Fr\'echet distance between two polygonal curves of $n$ vertices takes roughly quadratic time, and conditional lower bounds suggest that approximating to within a factor $3$ cannot be done in strongly-subquadratic time, even in one dimension.
    Currently, the best approximation algorithms present trade-offs between approximation quality and running time.
    At SoCG 2021, Colombe and Fox presented an $O((n^3 / \alpha^2) \log n)$-time $\alpha$-approximate algorithm for curves in arbitrary dimensions, for any $\alpha \in [\sqrt{n}, n]$.
    In this work, we give an $\alpha$-approximate algorithm with a significantly faster running time of $O((n^2 / \alpha) \log n)$, for any $\alpha \in [1, n]$.
    In particular, we give the first strongly-subquadratic $n^\eps$-approximation algorithm, for any constant $\varepsilon \in (0, 1/2]$.
    For curves in one dimension we further improve the running time to $O((n^2 / \alpha^3) \log^2 n)$, for $\alpha \in [1, n^{1/3}]$.
    Both of our algorithms rely on a linear-time simplification procedure that in one dimension reduces the complexity of the reachable free space to $O(n^2 / \alpha)$ without making sacrifices in the asymptotic approximation factor.
\end{abstract}

\section{Introduction}

    Comparing curves is an important task in for example trajectory analysis~\cite{chen05similarity_trajectories}, handwriting recognition~\cite{munich99signature_verification} and matching time series in data bases~\cite{liu13time_series_data_base}.
    To compare curves, one needs a suitable distance measure.
    The Hausdorff distance is a commonly used distance measure when comparing two shapes.
    However, although each curve corresponds to a shape (i.e., the image of the curve), a shape by itself does not capture the order in which points appear along the curve.
    This may lead to curves having low Hausdorff distance, even when they are clearly very different.
    The \f distance is a distance measure that does take the ordering of points along the curves into account, and hence compares curves more accurately.

    The first algorithm for computing the \f distance between polygonal curves was given by Godau~\cite{godau91continuous_frechet}, who presented an $O(n^3 \log n)$-time algorithm for two curves with $n$ vertices in total.
    Alt and Godau~\cite{alt95continuous_frechet} later improved the result to an $O(n^2 \log n)$-time algorithm.
    For curves on the real line that have an imbalanced number of vertices, i.e., $n$ and $n^\alpha$ for some $\alpha \in (0, 1)$, Blank and Driemel~\cite{blank24imbalanced_Frechet} recently gave a strongly subquadratic algorithm, taking $O(n^{2\alpha} \log^2 n + n \log n)$ time.
    The \f distance has a discrete analogue, introduced by Eiter and Mannila~\cite{eiter94discrete_frechet}.
    They gave an $O(n^2)$-time algorithm for this variant.

    The results for general dimensions and curves have since been improved in the word RAM model of computation.
    Agarwal~\etal~\cite{agarwal14discrete_frechet} gave an $O(n^2 \log \log n / \log n)$-time algorithm for the discrete problem, and Buchin~\etal~\cite{buchin17continuous_frechet} later improved the complexity bound for the continuous problem to $O(n^2 (\log \log n)^2)$.
    Unfortunately there is strong evidence that these results cannot be improved significantly, since Bringmann~\cite{bringmann14hardness} show that a \emph{strongly-subquadratic} (i.e., $n^{2-\Omega(1)}$-time) algorithm would refute the \emph{Strong Exponential Time Hypothesis} (SETH).
    
    Due to the conditional lower bound, we focus on efficient \emph{approximation algorithms}.
    When the curves are from certain families of ``realistic'' curves, strongly-subquadratic-time $(1+\eps)$-approximation algorithms are known to exist.
    For example, if the curves are either \emph{$\kappa$-bounded} or \emph{backbone} curves, the algorithm by Aronov~\etal~\cite{aranov06frechet_revisited} gives a $(1+\eps)$-approximation to the discrete \f distance in $O(n^{4/3} \log n / \eps^2)$ time.
    In the continuous setting, Driemel~\etal~\cite{driemel12realistic} give $(1+\eps)$-approximate algorithms that take near-linear time, given that the curves are from one of four realistic curve classes.
    These four classes include $\kappa$-bounded curves, but also \emph{$c$-packed}, \emph{$\varphi$-low density} and \emph{$\kappa$-straight} curves.
    Their result on $c$-packed curves was improved by Bringmann and K\"unnemann~\cite{bringmann17improved_cpacked}, whose algorithm matches conditional lower bounds.
    
    When approximating the \f distance between arbitrary curves, SETH again gives conditional lower bounds.
    The lower bound by Bringmann~\cite{bringmann14hardness} holds not only for exact algorithms, but for $1.001$-approximate algorithms as well.
    This lower bound was later improved by Buchin~\etal~\cite{buchin19seth_says}, who showed that under SETH, no strongly-subquadratic $(3-\eps)$-approximation algorithm exists, even for curves in one dimension.
    For the current strongly-subquadratic algorithms, the best known approximation factor is polynomial ($n^\eps$) for the discrete \f distance~\cite{bringmann16approx_discrete_frechet,chan18improved_approximation}.
    For the continuous \f distance, a (randomized) constant factor approximation algorithm with strongly subquadratic running time was recently posted on arXiv~\cite{cheng25constant_frechet}.
    
    For the discrete \f distance, Bringmann and Mulzer~\cite{bringmann16approx_discrete_frechet} give a linear time greedy algorithm with an approximation factor of $2^{\Theta(n)}$.
    They also present the first strongly-subquadratic-time algorithm with polynomial approximation factor.
    Their algorithm gives an $\alpha$-approximation in $O(n^2 / \alpha)$ time, for any $\alpha \in [1, n / \log n]$.
    This result was later improved by Chan and Rahmati~\cite{chan18improved_approximation}, who give an $O(n^2 / \alpha^2)$-time algorithm, for any $\alpha \in [1, \sqrt{n / \log n}]$.

    For continuous \f distance, Bringmann and Mulzer~\cite{bringmann16approx_discrete_frechet} again give a linear time greedy algorithm with an approximation factor of $2^{O(n)}$
    The first strongly-subquadratic time algorithm with polynomial approximation factor is due to Colombe and Fox~\cite{colombe21continuous_frechet}.
    They give an $\alpha$-approximate algorithm running in $O((n^3 / \alpha^2) \log n)$ time, for any $\alpha \in [\sqrt{n}, n]$.
    Their algorithm is strongly-subquadratic already for $\alpha = n^{1/2+\eps}$ for some constant $\eps > 0$.
    Very recently, Cheng~\etal~\cite{cheng25constant_frechet} gave the first (randomized) constant factor approximation algorithm with a strongly subquadratic running time.
    Specifically, it computes a $(7+\eps)$-approximation in $O(n^{1.99} \log (n/\eps))$ time.

\subparagraph*{Results.}
    We present new approximation algorithms for the continuous \f distance.
    The basis of our results is a curve simplification algorithm.
    We use the resulting simplified versions of two curves to efficiently approximate the Fr\'echet distance between the input curves.
    Doing so, we obtain a significantly improved tradeoff between approximation and running time compared to Colombe and Fox~\cite{colombe21continuous_frechet}.
    In arbitrary dimensions, our $\alpha$-approximation algorithm takes $O((n^2 / \alpha) \log n)$ time to complete.
    In one dimension, we improve this running time further to $O((n^2 / \alpha^3) \log^2 n)$.
    We summarize our algorithms in \cref{sec:algorithmic_outline}, but first we define the \f distance and some useful notation.

\section{Preliminaries}

    A $d$-dimensional (polygonal) \emph{curve} is a piecewise-linear function $P \from [0, 1] \to \R^d$, connecting a sequence $p_1, \dots, p_n$ of $d$-dimensional points, which we refer to as \emph{vertices}.
    The linear interpolation between $p_i$ and $p_{i+1}$, whose image is equal to the directed line segment $\overline{p_i p_{i+1}}$, is called an \emph{edge}.
    We denote by $P[x_1, x_2]$ the subcurve of $P$ over the domain $[x_1, x_2]$.
    We write $|P|$ to denote the number of vertices of $P$.

\subparagraph*{Fr\'echet distance.}
    A \emph{reparameterization} is a non-decreasing surjection $f \from [0, 1] \to [0, 1]$.
    Two reparameterizations $f, g$ describe a \emph{matching} $(f, g)$ between two curves $P$ and $Q$, where any point $P(f(t))$ is matched to $Q(g(t))$.
    A matching $(f, g)$ between $P$ and $Q$ is said to have \emph{cost}
    \[
        \max_t~\lVert P(f(t)) - Q(g(t)) \rVert.
    \]
    It is common to use the Euclidean norm $\lVert P(f(t)) - Q(g(t)) \rVert_2$ to measure the cost of a matching.
    For our purposes however, it is more convenient to use the $L_\infty$ norm $\lVert P(f(t)) - Q(g(t)) \rVert_\infty$.
    Since we aim for at least polynomial approximation factors, and the norms differ by at most a factor $\sqrt{d}$, approximations using the $L_\infty$ norm imply the same asymptotic approximation factor for the Euclidean norm, as long as $d$ is constant.
    A matching with cost at most $\delta$ is called a \emph{$\delta$-matching}.
    The (continuous) \emph{\f distance} $d_F(P, Q)$ between $P$ and $Q$ is the minimum cost over all matchings.

\subparagraph*{Free space diagram and matchings.}
    The \emph{free space diagram} of $P$ and $Q$ is the parameter space~$[0, 1]^2$ of $P\times Q$, denoted $\D(P, Q)$.
    Any point $(x, y) \in \D(P, Q)$ corresponds to the pair of points $P(x)$ and $Q(y)$ on the two curves.
    Any pair of edges $(P[x_1, x_2],Q[y_1, y_2])$ corresponds to a \emph{cell}~$[x_1, x_2] \times [y_1, y_2]$ of $\D(P, Q)$.
    
    For $\delta \geq 0$, a point $(x, y) \in \D(P, Q)$ is \emph{$\delta$-close} if $\lVert P(x) - Q(y) \rVert_\infty \leq \delta$.
    The \emph{$\delta$-free space} $\F(P, Q)$ of $P$ and $Q$ is the subset of $\D(P, Q)$ containing all $\delta$-close points.
    A point $z_2 = (x_2, y_2) \in \F(P, Q)$ is \emph{$\delta$-reachable} from a point $z_1 = (x_1, y_1)$ with $x_1 \leq x_2$ and $y_1 \leq y_2$ if there exists a bimonotone path in $\F(P, Q)$ from $z_1$ to $z_2$.
    Points that are $\delta$-reachable from $(0, 0)$ are simply called $\delta$-reachable points.
    Alt and Godau~\cite{alt95continuous_frechet} observe that the \f distance between $P[x_1, x_2]$ and $Q[y_1, y_2]$ is at most $\delta$ if and only if the point $(x_2, y_2)$ is $\delta$-reachable from $(x_1, y_1)$.
    We therefore abuse terminology slightly and refer to a bimonotone path from $z_1$ to $z_2$ as a $\delta$-matching between $P[x_1, x_2]$ and $Q[y_1, y_2]$.

\section{Algorithmic outline}
\label{sec:algorithmic_outline}

    \begin{figure}
        \centering
        \includegraphics{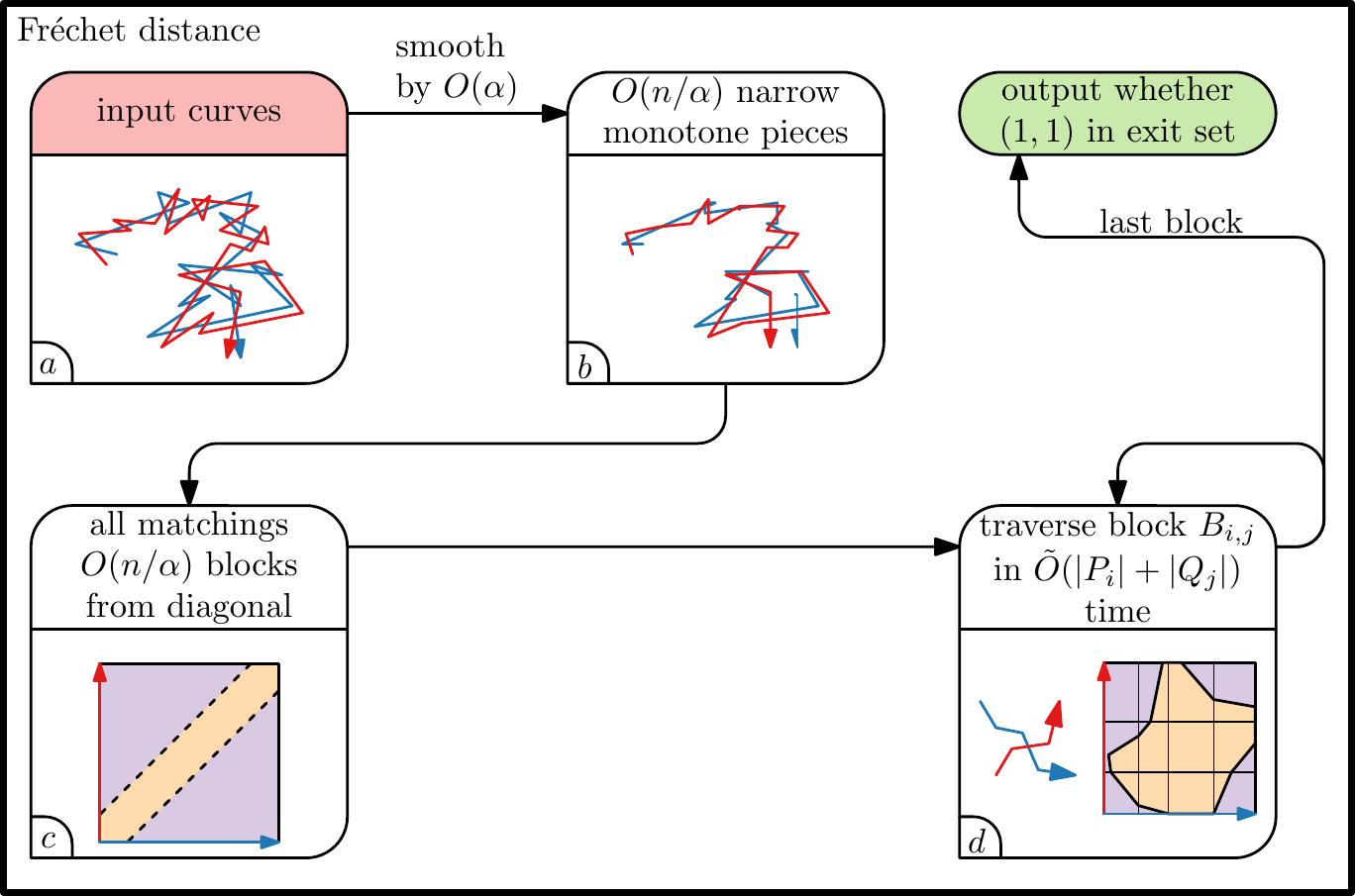}
        \caption{An illustration of the approximate decision algorithm.}
        \label{fig:diagram_all_dimensions}
    \end{figure}
    
    Let $P$ and $Q$ be our two $d$-dimensional input curves with $n$ vertices.
    Given a parameter $\alpha \in [1, n]$, we describe an \emph{$\alpha$-approximate decision algorithm} for the continuous \f distance.
    Such an algorithm takes as input an additional parameter $\delta \geq 0$, and must correctly report that $d_F(P, Q) \leq \alpha \delta$ or that $d_F(P, Q) > \delta$.
    If $d_F(P, Q) \in (\delta, \alpha \delta]$, the algorithm may report either.
    We thus either confirm that an $\alpha \delta$-matching exists, or assert that no $\delta$-matching exists.
    Refer to \cref{fig:diagram_all_dimensions} for a diagram illustrating our algorithm.
    We turn our decision algorithms into approximation algorithms for the \f distance with the procedure of Colombe and Fox~\cite{colombe21continuous_frechet} (with logarithmic overhead in the running time and arbitrarily small increase in approximation ratio).
    
    Recall that a $\delta$-matching between $P$ and $Q$ represents a bimonotone path from $(0, 0)$ to $(1, 1)$ in the $\delta$-free space $\F(P, Q)$.
    Our algorithms search for such a path.
    However, exploring all of the free space, which may have $\Theta(n^2)$ complexity, does not result in a subquadratic-time algorithm.
    Still, the worst-case complexity of the \emph{reachable} free space, the part of free space containing all $\delta$-reachable points, is smaller for certain types of input curves $P$ and $Q$.
    We explore this in~\cref{sec:free_space_complexity}, where we investigate the relation between the complexity of the reachable free space for one-dimensional curves, and the number of edges that are short with respect to $\delta$.
    If $P$ and $Q$ have $k$ short edges, then the complexity of their reachable free space is only $O(kn)$.

    Additionally in~\cref{sec:free_space_complexity}, we generalize the obtained bound on the reachable free space complexity to hold for curves in higher dimensions.
    There, we consider the \emph{monotone pieces} of a curve; the maximum subcurves that are monotone in all coordinates.
    If the projections of the input curves have $k$ short edges each, then the reachable free space complexity is only $O(kn)$ \emph{blocks}.
    Here, a block is a generalization of cells that instead of edges considers monotone pieces.
    Specifically, a block is the rectangular region of the free space diagram corresponding to two monotone pieces, one of $P$ and one of~$Q$.
    
    Given that a sublinear number of short edges in the projections implies a subquadratic complexity of the reachable free space, we present a simplification procedure in \cref{sec:reducing_narrow} that reduces the number of short edges of a one-dimensional curve to at most $n / \alpha$, at an additive factor of $2\alpha$ to the approximation ratio.
    This simplification is applied to the projections of the input curves, see \cref{fig:diagram_all_dimensions} (a--b).
    The simplification takes linear time, and results in a reachable free space complexity of only $O(n^2 / \alpha)$ blocks, see \cref{fig:diagram_all_dimensions} (c).
    Intuitively, the proportion of the free space diagram that we need to explore is inversely proportional to the approximation factor.
    
    A block $B_{i, j}$ consists of $O(|P_i| \cdot |Q_j|)$ cells, each of which may have a non-empty intersection with the free space.
    Even when considering only the $O(n^2 / \alpha)$ blocks containing the reachable free space, the number of cells containing reachable points may be $\Theta(n^2)$.
    However, the free space inside a block is ortho-convex,\footnote{
        A region $S$ is ortho-convex if every line parallel to a coordinate axis intersects $S$ in at most one connected component.
        Note that $S$ does not need to be connected.
    }
    see \cref{fig:diagram_all_dimensions} (d).
    We use this fact in \cref{sec:alternative_algorithm} to show that the boundary of the free space inside a block $B_{i, j}$ has only $O(|P_i| + |Q_j|)$ complexity, which allows us to propagate reachability information through $B_{i, j}$ in $O(|P_i| + |Q_j|)$ time.
    This gives an $O(n^2 / \alpha)$-time algorithm for propagating reachability information through all $O(n^2 / \alpha)$ blocks, and hence gives a $(2\alpha + 1)$-approximate decision algorithm with the same running time.
    The technique by Colombe and Fox~\cite{colombe21continuous_frechet} can be applied to turn the decision algorithm into an $\alpha$-approximation algorithm for the \f distance, with a running time of $O((n^2 / \alpha) \log n)$.
    
    \begin{figure}
        \centering
        \includegraphics{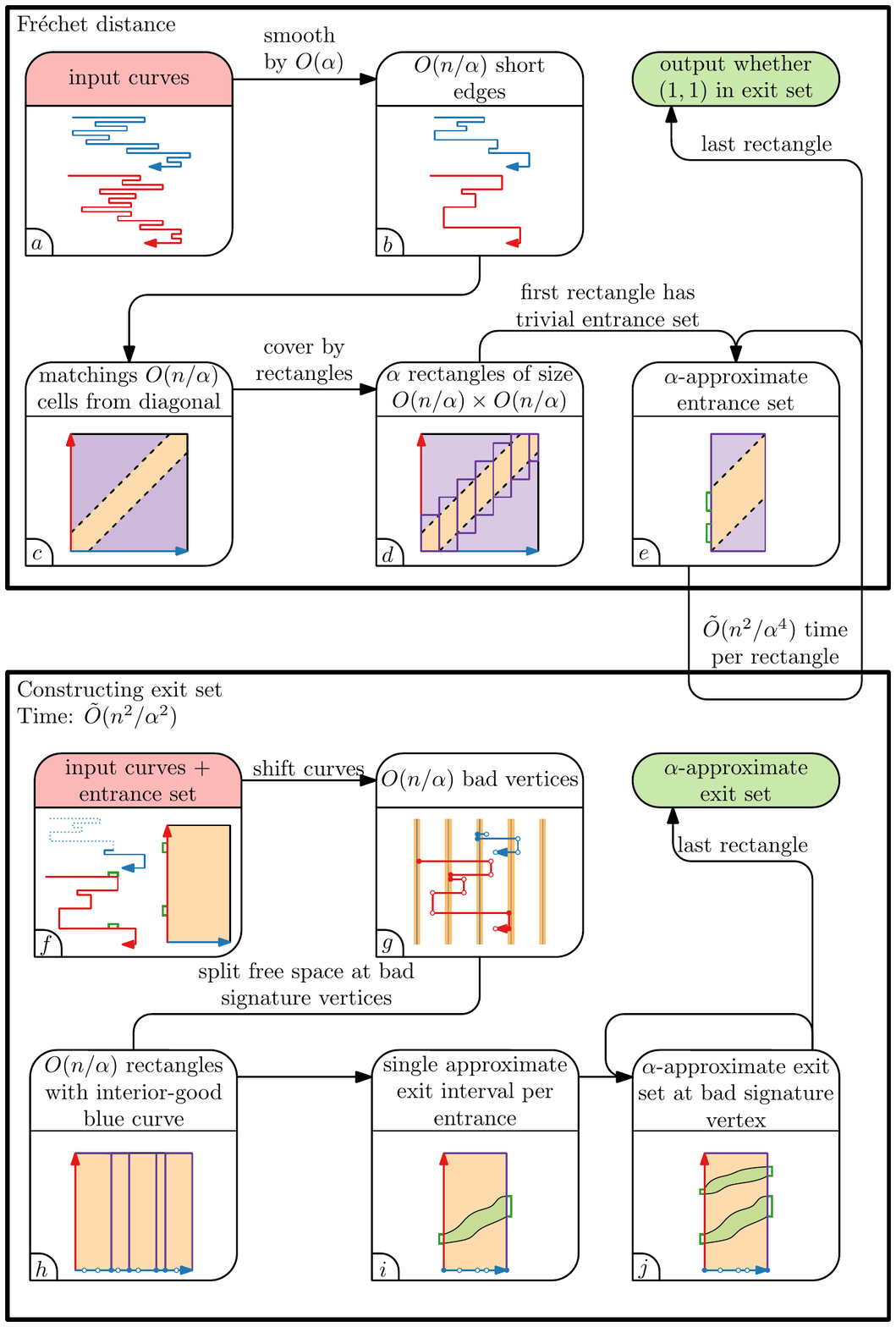}
        \caption{An illustration of the approximate decision algorithm for one-dimensional curves.
        Our main algorithm (top diagram) repeatedly calls a subroutine for constructing approximate exit sets (bottom diagram).}
        \label{fig:diagram_one_dimension}
    \end{figure}
    
    In \cref{sec:faster_one_dimension,sec:exit_sets} we give a faster algorithm for when $P$ and $Q$ are one-dimensional curves.
    \Cref{fig:diagram_one_dimension} illustrates this algorithm.
    The core of the algorithm is a subroutine for constructing \emph{approximate exit sets} (see the bottom diagram of \cref{fig:diagram_one_dimension}).
    Given a set of points $S \subseteq \{0\} \times [0, 1]$ on the left side of the free space diagram, an $(\alpha, \delta)$-exit set for $S$ is a set of points $E_\alpha(S) \subseteq \{1\} \times [0, 1]$ on the right side of the diagram that contains all points that are $\delta$-reachable from $S$, and only points that are $\alpha \delta$-reachable from $S$.
    If $(1, 1) \in E_\alpha(\{(0, 0)\})$, then $d_F(P, Q) \leq \alpha \delta$, and otherwise $d_F(P, Q) > \delta$.
    Computing such exit sets thus generalizes the approximate decision problem.
    
    We construct approximate exit sets in \cref{sec:exit_sets}.
    For this we use the ideas of Chan and Rahmati~\cite{chan18improved_approximation} for the current state-of-the-art discrete approximate decision algorithm.
    They construct a graph approximately representing the free space, which can be used to construct approximate exit sets (in the discrete setting).
    These exit sets take only $O(n \log n + n^2 / \alpha^2)$ time to construct.
    
    To achieve a similar running time in the continuous setting, we first note that continuous matchings in one dimension are relatively discrete.
    In particular, \emph{signature vertices}, special vertices introduced by Driemel~\etal~\cite{driemel15clustering}, must match in an almost discrete manner, matching to points close to vertices of the other curve.
    With this in mind, we apply the techniques of Chan and Rahmati~\cite{chan18improved_approximation} to the signature vertices of $P$.
    
    We construct an infinite grid $\G$ with few \emph{bad} vertices of both $P$ and $Q$.
    See \cref{fig:diagram_one_dimension} (f--g).
    This grid has cellwidth $\alpha \delta$, and we classify a point as bad if it is within distance $7\delta$ of the boundary of $\G$.
    By shifting $P$ and $Q$, the number of bad vertices can be made as low as $O(n / \alpha)$.
    Moreover, such a shift can be computed in $O(n)$ time~\cite{chan18improved_approximation}.
    We say that a signature vertex of $P$ is bad if it is within distance $6\delta$ of the boundary of $\G$, rather than within distance $7\delta$.
    These vertices must match to points close to bad vertices of $Q$, and hence have only $O(n / \alpha)$ possible ways to match to points.
    
    Between two bad signature vertices of $P$, the signature vertices are all sufficiently far from the boundary of $\G$ that we can represent them by the grid cells containing them, after which matchings become effectively diagonal.
    We can detect such matchings with the exact string matching data structure by Chan and Rahmati~\cite{chan18improved_approximation}, and use an additional data structure to handle the matchings around bad signature vertices.
    For a single entrance, we can then efficiently compute an $(\alpha, \delta)$-exit set for any subcurve between two subsequent bad signature vertices, see \cref{fig:diagram_one_dimension} (h--i).
    The data structure constructs such a set in only $O(\log n)$ time.
    Applied to all $O(n^2 / \alpha^2)$ possible matchings with a bad signature vertex of $P$, we get an $O((n^2 / \alpha^2) \log n)$-time algorithm for constructing $(\alpha, \delta)$-exit sets of general sets of points, after $O(n \log n)$ preprocessing time.
    
    We can improve the above algorithm by taking advantage of the lower-complexity reachable free space.
    Given that the reachable free space stays within $O(n / \alpha)$ cells of the diagonal, we cover this region by $\alpha$ rectangles of size $O(n / \alpha) \times O(n / \alpha)$ cells.
    See \cref{fig:diagram_one_dimension} (c--d).
    In each rectangle we construct an $(\alpha, \delta)$-exit set for a given set of entrance points, which depend on the exit set of the rectangle to the left of the current one.
    These exit sets take only $O((n^2 / \alpha^4) \log n)$ time to construct for a rectangle, totalling $O((n^2 / \alpha^3) \log n)$ time.
    This is a factor $\alpha$ improvement, which we would expect given the lower complexity of the reachable free space.
    
    The technique by Colombe and Fox~\cite{colombe21continuous_frechet} turns our algorithm for constructing a $(\alpha, \delta)$-exit set into an $O((n^2 / \alpha^3) \log^2 n)$-time $\alpha$-approximate algorithm for the \f distance in one dimension.

\section{Bounding the reachable free space complexity}
\label{sec:free_space_complexity}

    The complexity of the $\delta$-free space can be as high as $\Theta(n^2)$, meaning that explicitly traversing the free space does not give a strongly-subquadratic-time algorithm.
    As an improvement, we aim to bound the complexity of the \emph{reachable $\delta$-free space}, the subset of $\delta$-free space that is reachable by a bimonotone path from $(0, 0)$.
    This subset contains all bimonotone paths to $(1, 1)$, so it suffices to consider only this subset.
    
    Like the complexity of free space, the complexity of the reachable free space can be quadratic.
    Still, there are special cases of curves for which we can check if the top-right point $(1, 1)$ is reachable in as little as linear time.
    For example, Gudmundsson~\etal~\cite{gudmundsson19long} show that for one-dimensional curves, if both curves have long edges only, i.e. all edges have length greater than $2\delta$, then the reachable free space is traversable in linear time.
    See~\cref{fig:long_edges_free_space}.
    Their result extends to higher dimensions.

    \begin{figure}[b]
        \centering
        \includegraphics{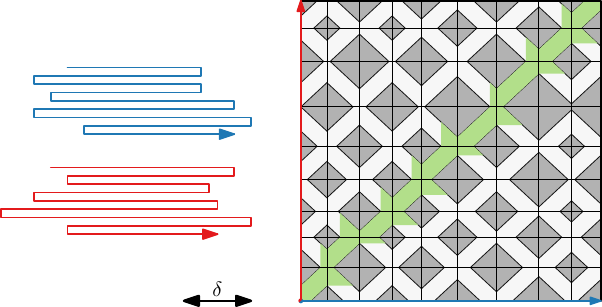}
        \caption{(left) Two one-dimensional curves (with vertices replaced by vertical segments for clarity) that have only long edges with respect to $\delta$.
        (right) The $\delta$-free space (white) and all $\delta$-reachable points (green).
        The reachable points all lie in blocks (cells) close to the diagonal.}
        \label{fig:long_edges_free_space}
    \end{figure}
    
    In this section, we determine a relation between the complexity of the reachable free space and the shapes of the projections of $P$ and $Q$.
    In particular, we show that if the $d$ projections of each curve onto the $d$ coordinate axes all have long edges only, then their reachable $\delta$-free space has linear complexity in some sense.
    This somewhat generalizes the result of Gudmundsson~\etal~\cite{gudmundsson19long}: the projections of a curve with sufficiently long edges have only long edges as well.

    We analyze the complexity of the reachable $\delta$-free space in terms of a generalization of cells.
    Namely, we analyze the complexity in terms of blocks of cells that correspond to monotone pieces of the curves.
    A curve is monotone if in every coordinate it is non-increasing or non-decreasing.
    The \emph{monotone pieces} $P_1, \dots, P_k$ of $P$ are the maximal monotone subcurves $P_i = P[x_i, x_{i+1}]$ that start at the endpoint of the last monotone piece.
    Note that these monotone pieces are not maximal with respect to $P$.    
    The monotone pieces of $Q$ are defined symmetrically.
    
    We define the \emph{block} $B_{i, j} \coloneqq [x_i, x_{i+1}] \times [y_j, y_{j+1}]$ to be the subset of $\D(P, Q)$ corresponding to the $i^{\mathrm{th}}$ monotone pieces $P_i = P[x_i, x_{i+1}]$ of $P$ and the $j^{\mathrm{th}}$ monotone piece $Q_j = Q[y_j, y_{j+1}]$ of $Q$.
    This block is the union of the cells defined by the edges of $P_i$ and~$Q_j$.
    
    Under the $L_\infty$ norm, the monotone pieces of a curve behave much like line segments.
    Most importantly, any ball under the $L_\infty$ norm intersects a piece in at most one connected component.
    For the free space, this implies that each block $B_{i, j}$ has an ortho-convex intersection with the $\delta$-free space, for all $\delta$.
    This somewhat generalizes the convexity of the free space within a cell (defined by two line segments) to blocks.

    For our analysis of the complexity of the reachable $\delta$-free space, we first consider the one-dimensional case.
    Here, a monotone piece is essentially a subdivided edge; it is a subdivided directed line segment between two local extrema.
    We say that a monotone piece is \emph{$2\delta$-short} if its arc-length is at most $2\delta$.
    In~\cref{thm:narrow_vs_reachable_1D}, we generalize the one-dimensional result of Gudmundsson~\etal~\cite{gudmundsson19long} to obtain a relation between the number of blocks containing the reachable $\delta$-free space, and $2\delta$-short monotone pieces.
    
    \begin{wrapfigure}{r}{0.31\textwidth}
        \centering
        \vspace{-11pt}
        \includegraphics{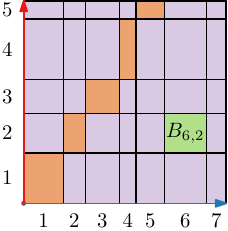}
        \vspace{-11pt}
    \end{wrapfigure}
    
    We obtain some more structure to the reachable free space, namely that all reachable points are somewhat close to the \emph{block diagonal}.
    The block diagonal is the set of blocks $B_{i, i}$, and we say that $B_{i, j}$ has distance $|i-j|$ to the block diagonal.
    For example, in the figure on the right, block $B_{6, 2}$ is $4$ blocks away from the block diagonal (orange).
    The additional structure on the reachable free space proves useful in~\cref{sec:faster_one_dimension}, where we use it in our algorithm for approximating the \f distance in one dimension.

    \begin{theorem}
    \label{thm:narrow_vs_reachable_1D}
        Let $P$ and $Q$ be one-dimensional curves that each have at most $k$ monotone pieces that are $2\delta$-short.
        The reachable $\delta$-free space is contained in the blocks within distance $2k+1$ of the block diagonal.
    \end{theorem}
    \begin{proof}
        We prove that there are no blocks with distance greater than $2k+1$ of the block diagonal that contain a $\delta$-reachable point.
        We prove the claim for such blocks to the right of the diagonal; the proof for the other blocks is symmetric.
        
        We classify each block as either \emph{equally-oriented} or \emph{oppositely-oriented}.
        Block $B_{i, j}$ is equally-oriented if its corresponding pieces, which are both subdivided directed segments, have the same direction, and oppositely-oriented otherwise.
        This classification, together with the fact that in one dimension, two consecutive monotone pieces of a curve have opposite directions, gives rise to a ``checkerboard pattern'' in the free space diagram, where each equally-oriented block shares sides with only oppositely-oriented blocks, and vice versa.
        We make use of the checkerboard pattern by looking at the sets $\mathcal{B}_i = \{B_{i+j-1, j} \mid 1 \leq j \leq n-i+1\}$ for $i \geq 1$, which we classify as either equally-oriented or oppositely-oriented, based on the classification of the contained blocks.

        Consider a bimonotone path $\pi$ in $\F(P, Q)$ that starts at $(0, 0)$.
        We show that there are at most $k$ oppositely-oriented sets $\mathcal{B}_i$ that $\pi$ traverses horizontally, i.e., where $\pi$ reaches a block in $\mathcal{B}_{i+1}$.
        Let $B_{i, j}$ be an oppositely-oriented block, and let $P_i$ and $Q_j$ be its corresponding monotone pieces (of $P$ and $Q$, respectively).
        Observe that if $P_i$ is long, i.e., has length greater than $2\delta$, then there exists no $\delta$-matching from a point on the left side of $B_{i, j}$ to a point on the right side.
        Indeed, suppose for a contradiction that there is such a path.
        Then there would be a $\delta$-matching between $P_i$ and some directed line segment $Q_j[y, y']$, which is a subcurve of $Q_j$.
        However, this means that $|P_i(0) - Q_j(y)| \leq \delta$ and $|P_i(1) - Q_j(y')| \leq \delta$.
        Due to the opposing directions of $P_i$ and $Q_j$, this is not possible unless $P_i$ has length at most $2\delta$, giving a contradiction.

        By the above, there are at most $k$ oppositely-oriented sets $\mathcal{B}_i$ that $\pi$ traverses horizontally.
        As we increase $i$, the sets $\mathcal{B}_i$ alternate between being oppositely-oriented and equally-oriented.
        It therefore follows that $\pi$ reaches blocks in at most $2k+1$ sets, which are within $2k+1$ blocks of the block diagonal.
    \end{proof}

    We extend the result of~\cref{thm:narrow_vs_reachable_1D} to obtain a complexity bound when $P$ and $Q$ are higher-dimensional curves.
    We show that the number of blocks containing the reachable $\delta$-free space depends linearly on the number of short monotone pieces of the projections of $P$ and $Q$.

    \begin{theorem}
    \label{thm:narrow_vs_reachable}
        Let $P$ and $Q$ be $d$-dimensional curves, whose projections onto the $d$ coordinate axes all have at most $k$ monotone pieces that are $2\delta$-short.
        The reachable $\delta$-free space is contained in $O(kn)$ blocks.
        Moreover, in a single row or column of blocks, there are at most $O(k)$ blocks containing a $\delta$-reachable point.
    \end{theorem}
    \begin{proof}
        We consider sparser subdivisions of the parameter space into blocks.
        Namely, for each dimension $\ell = 1, \dots, d$, we define an \emph{$\ell$-unimonotone block} $B_{i, j}^\ell \coloneqq [x_i^\ell, x_{i+1}^\ell] \times [y_j^\ell, y_{j+1}^\ell]$ to be the subset of $\D(P, Q)$ corresponding to maximal subcurves $P[x_i^\ell, x_{i+1}^\ell]$ and $Q[y_j^\ell, y_{j+1}^\ell]$ that are monotone in their $\ell^{\mathrm{th}}$ coordinate.

        For any dimension $\ell$, let $P_\ell$ and $Q_\ell$ be the projections of $P$ and $Q$ onto the $\ell^{\mathrm{th}}$ coordinate axis.
        Since it is given that these projections each have at most $k$ monotone pieces that are $2\delta$-short, it follows from~\cref{thm:narrow_vs_reachable_1D} that the reachable subset of $\F(P_\ell, Q_\ell)$ is contained in the blocks within distance $2k+1$ of the cell diagonal (in $\D(P_\ell, Q_\ell)$).
        Hence each row or column of cells contains at most $O(k)$ blocks with a reachable point of $\F(P_\ell, Q_\ell)$.

        The above implies that each row or column of $\ell$-unimonotone blocks contains at most $O(k)$ such blocks with a reachable point of $\F(P_\ell, Q_\ell)$.
        By our use of the $L_\infty$ norm we have $\F(P, Q) \subseteq \F(P_\ell, Q_\ell)$.
        Since each block of $\D(P, Q)$ is contained inside an $\ell$-unimonotone block, for some $\ell$, it follows that each row or column of blocks contains at most $O(dk)$ blocks with a reachable point of $\F(P, Q)$.
    \end{proof}

\section{Simplifying the projections}
\label{sec:reducing_narrow}
    
    We present a family of simplifications for curves that we use to reduce the number of short edges the projections have.
    The simplifications are based on truncated smoothings for Reeb graphs~\cite{chambers21smoothing}, and we hence call them truncated smoothings.
    The simplifications operate on the projections themselves.
    Thus, throughout this section, we consider simplifying a one-dimensional curve $P$ with $n$ vertices.

    We abuse terminology slightly and refer to the monotone pieces of $P$ as its edges.
    Vertices that are not endpoints of monotone pieces, and thus not endpoints of these edges, are called \emph{degenerate}.
    Although there exists a curve with no degenerate vertices and \f distance $0$ to $P$, we explicitly track degenerate vertices.
    This is done so the simplifications can be applied to higher-dimensional curves, see~\cref{sec:alternative_algorithm}.

\subsection{Truncated smoothings}
\label{sub:truncated_smoothing}

    Let $\eps \geq 0$ be at most half the minimum edge length of $P$.
    The \emph{truncated $\eps$-smoothing} $P^\eps$ of $P$ is the curve obtained by truncating every edge of $P$ by $\eps$ on either side.
    See~\cref{fig:smoothing_example} for an example.
    Formally, $P^\eps$ is a curve with $n$ vertices, where a point $P^\eps(x)$ has the following value:
    If $P(x)$ lies on an edge $\overline{p_i p_j}$, then
    \[
        P^\eps(x) = \begin{cases}
            p_i & \text{if $|P(x) - p_i| \leq \eps$,}\\
            p_j & \text{if $|P(x) - p_j| \leq \eps$,}\\
            P(x) & \text{otherwise.}
        \end{cases}
    \]
    
    Let $\eps'$ be half the minimum edge length of $P$.
    We extend the truncated smoothing definition to all non-negative values $\eps \geq 0$ by recursively defining the truncated $\eps$-smoothing $P^\eps$ of $P$, for values $\eps > \eps'$, to be the truncated $(\eps - \eps')$-smoothing of $P^{\eps'}$.
    When $\eps' = 0$, we instead let $P^\eps = P$ for all $\eps \geq 0$.

    \begin{remark*}
        Truncated smoothings are related to the concept of \emph{signatures} of a curve, introduced by~\cite{driemel15clustering}.
        Roughly, the vertices of $P$ that do not become degenerate in $P^\eps$, are the $\eps$-signature vertices of $P$.
        For a definition of signatures, see~\cref{sec:exit_sets}, where we make extensive use of their properties.
    \end{remark*}
    
    \begin{figure}
        \centering
        \includegraphics{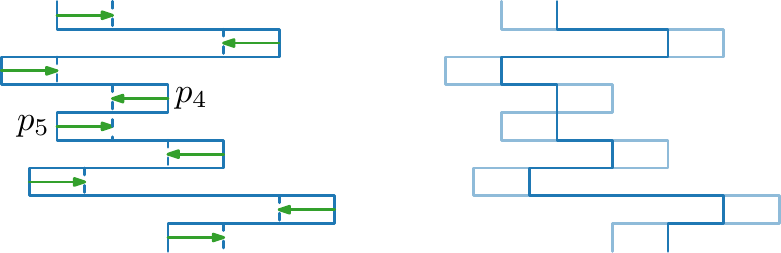}
        \caption{An illustration of truncated smoothings.
        (left) The vertices of curve $P$ (non-dashed) are drawn as vertical segments for clarity.
        The minimum edge length of $P$ is realized by $\overline{p_i p_j}$.
        (right) The result of the smoothing procedure (opaque).
        }
        \label{fig:smoothing_example}
    \end{figure}

    \begin{lemma}
    \label{lem:smoothing_matching}
        Let $P$ be a one-dimensional curve.
        We have $|P(x) - P^\eps(x)| \leq \eps$ for all $x \in [0, 1]$ and $\eps \geq 0$.
    \end{lemma}
    \begin{proof}
        Let $\eps \geq 0$ and let $\eps'$ be half the minimum edge length of $P$.
        Consider a point $P(x)$ on $P$.
        If $\eps \leq \eps'$ or $\eps' = 0$, then the point $P^\eps(x)$ is trivially within distance $\eps$ of $P(x)$.
        For general $\eps$ and positive $\eps'$, applying the triangle inequality to the recursive definition of the simplification yields that
        \[
            | P(x) - P^\eps(x) | \leq \begin{cases}
                \eps & \text{if $\eps \leq \eps'$,}\\
                | P(x) - P^{\eps - \eps'}(x) | + \eps' & \text{otherwise}.
            \end{cases}
        \]
        This implies $| P(x) - P^\eps(x) | \leq \eps$.
    \end{proof}

    We prove that the truncated smoothings do not increase the \f distance between the curves.
    In particular, the reparameterizations of any $\delta$-matching between two curves specify a $\delta$-matching between their simplifications.
    This fact becomes useful in higher dimensions (see~\cref{thm:reducing_short}), where we compute truncated smoothings coordinate-wise, and wish to use a single matching between $d$-dimensional curves to specify $d$ matchings between pairs of truncated smoothings.

    \begin{lemma}
    \label{lem:smoothing_error_1D}
        Let $P$ and $Q$ be one-dimensional curves and let $\delta \geq 0$.
        For all $\eps \geq 0$, a $\delta$-matching $(f, g)$ between $P$ and $Q$ induces a $\delta$-matching between $P^\eps$ and $Q^\eps$, given by the same reparameterizations $f$ and $g$.
        Conversely, a $\delta$-matching $(f, g)$ between $P^\eps$ and $Q^\eps$ induces a $(\delta+2\eps)$-matching between $P$ and $Q$, given by $f$ and $g$ as well.
    \end{lemma}
    \begin{proof}
        Let $(f, g)$ be a $\delta$-matching between $P^\eps$ and $Q^\eps$.
        It follows from~\cref{lem:smoothing_matching} and the triangle inequality that $|P(f(t)) - Q(g(t))| \leq |P^\eps(f(t)) - Q^\eps(g(t))| + 2\eps$.
        Hence the matching between $P$ and $Q$ induced by $(f, g)$ has cost at most $\delta+2\eps$.

        Next we prove that a $\delta$-matching $(f, g)$ between $P$ and $Q$ induces a $\delta$-matching between $P^\eps$ and $Q^\eps$, given by the same reparameterizations $f$ and $g$.
        If $\delta = 0$, then $P(f(t)) = Q(g(t))$ for all $t \in [0, 1]$, and so trivially $P^\eps(f(t)) = Q^\eps(g(t))$ as well.
        In the remainder we assume $\delta > 0$.

        To prove the main statement, we make use of two claims.
        The first is for the case where the minimum edge length of $P$ and $Q$ is $0$, so the truncated $\eps$-smoothing of one of the curves is equal to the curve itself.

        \begin{claim}
        \label{clm:smoothing_error_min_length_0}
            If the minimum edge length of $P$ and $Q$ is $0$, then the reparameterizations $f$ and $g$ specify a $\delta$-matching between $P^\eps$ and $Q^\eps$ for all $\eps \geq 0$.
        \end{claim}
        \begin{claimproof}
            The image of $P^\eps$ (the smallest interval containing $P^\eps$) lies in the image of $P$ for all $\eps \geq 0$.
            Symmetrically, the image of $Q^\eps$ lies in the image of $Q$ for all $\eps \geq 0$.
            If $P$ has a minimum edge length of $0$, the images of all $P^\eps$ are equal to some point $p$.
            The cost of $(f, g)$ is therefore equal to the maximum distance from $p$ to a point on $Q$.
            As $\eps$ increases from $0$, the maximum distance from $p$ to a point on $Q^\eps$ decreases, since the image of $Q^\eps$ shrinks.
            It follows that the matching between $P^\eps$ and $Q^\eps$ induced by $(f, g)$ has cost at most $\delta$.
            The same holds for when $Q$ has a minimum edge length of $0$.
        \end{claimproof}

        The second claim proves that if the minimum edge length is positive, then there exists a range of parameters $\eps$ for which the reparameterizations $f$ and $g$ specify a $\delta$-matching between $P^\eps$ and $Q^\eps$.

        \begin{claim}
        \label{clm:smoothing_error_small_parameters}
            Let $\eps'$ be the half the minimum edge length of $P$ and $Q$.
            If $\eps' > 0$, then the reparameterizations $f$ and $g$ specify a $\delta$-matching between $P^\eps$ and $Q^\eps$ for all $\eps \in (0, \min\{\eps', \delta/2\}]$.
        \end{claim}
        \begin{claimproof}
            Assume for a contradiction that there exists a value $\eps \in (0, \min\{\eps', \delta/2\}]$ such that the matching between $P^\eps$ and $Q^\eps$, specified by the reparameterizations $f$ and $g$, has cost greater than $\delta$.
            Then there exists a value $t \in [0, 1]$ such that $|P(f(t)) - Q(g(t))| \leq \delta$ and $|P^\eps(f(t)) - Q^\eps(g(t))| > \delta$.
    
            Let $x = f(t)$ and $y = g(t)$.
            We assume without loss of generality that $P(x) \leq Q(y)$.
            We first handle the case $P^\eps(x) < P(x)$, and show afterwards that the case $P^\eps(x) \geq P(x)$ is symmetric.

            When $P^\eps(x) < P(x)$, then because $\eps$ is at most half the minimum edge length, $P(x)$ lies on an edge that has an endpoint with value $P^\eps(x)+\eps$.
            Hence there exists a subcurve $P' = P[x_1, x_2]$, with $x_1 \leq x \leq x_2$, that starts and ends at the value $P^\eps(x)$ and whose image is $[P^\eps(x), P^\eps(x)+\eps]$.
            Also, because $\eps \leq \eps'$, the point $Q(y)$ must lie on an edge of length at least $2\eps$, and so this edge has an endpoint in the interval $[Q^\eps(y)+\eps, \infty)$.
            Hence there exists a subcurve $Q' = Q[y_1, y_2]$, with $y_1 \leq y \leq y_2$, that starts or ends at $Q^\eps(y)+\eps$ and whose image is the interval $[Q(y), Q^\eps(y)+\eps]$.
    
            Both endpoints of $P'$ have distance greater than $\delta$ to all points on $Q'$.
            Since $(f, g)$ matches $P(x)$ to $Q(y)$, it must match $P(x_1)$ to before $Q(y_1)$ and match $P(x_2)$ to after $Q(y_2)$.
            Therefore, both endpoints of $Q'$ must be matched to points on $P'$.
            However, the endpoint of $Q'$ with value $Q^\eps(y)+\eps$ has distance greater than $\delta$ to all points on $P'$.
            Thus, the cost of $(f, g)$ is greater than $\delta$, giving a contradiction.

            We complete the proof by showing that the case $P^\eps(x) \geq P(x)$ is symmetric to the above case.
            When $P^\eps(x) \geq P(x)$, then $P^\eps(x) \in [P(x), P(x)+\eps]$ and $Q^\eps(y) \in [Q(y)-\eps, Q(y)+\eps]$, by~\cref{lem:smoothing_matching}.
            Because $\eps \leq \delta/2$, it must be that $Q^\eps(y) > Q(y)$, as otherwise $|P^\eps(x) - Q^\eps(y)| \leq \delta$.
            This situation is symmetric to the above, with the roles of $P$ and $Q$ swapped.
            Thus, we also get a contradiction when $P^\eps(x) \geq P(x)$.
        \end{claimproof}

        We now prove the main statement.
        Suppose for a contradiction that there exists a maximum parameter $\eps^* < \infty$ for which $f$ and $g$ specify a $\delta$-matching between $P^{\eps^*}$ and $Q^{\eps^*}$.
        It follows from the definition of truncated smoothings that for any $\eps > 0$, the truncated $(\eps^* + \eps)$-smoothings of $P$ and $Q$ are the truncated $\eps$-smoothings of $P^{\eps^*}$ and $Q^{\eps^*}$, respectively.
        Let $\eps'$ be the minimum edge length of $P^{\eps^*}$ and $Q^{\eps^*}$.
        We obtain from~\cref{clm:smoothing_error_min_length_0} that if $\eps' = 0$, then $f$ and $g$ specify a $\delta$-matching between $P^{\eps^* + \eps}$ and $Q^{\eps^* + \eps}$ for every $\eps \geq 0$.
        Also, when $\eps' > 0$, then by~\cref{clm:smoothing_error_small_parameters} we have that $f$ and $g$ specify a $\delta$-matching between $P^{\eps^* + \eps}$ and $Q^{\eps^* + \eps})$ for every $\eps \in (0, \min\{\eps', \delta/2\}]$.
        Since the half-open interval $(0, \min\{\eps', \delta/2\}]$ is non-empty, it follows that in either case, there exist a range of parameters $\eps > 0$ for which $f$ and $g$ specify a $\delta$-matching between $P^{\eps^* + \eps}$ and $Q^{\eps^* + \eps}$.
        This contradicts the fact that $\eps^*$ is maximum, completing the proof.
    \end{proof}

    Next we give a crucial property of truncated smoothings, showing that we can bound the number of edges of any given length to any given number with some truncated smoothing:

    \begin{lemma}
    \label{lem:smoothing_1D}
        Let $P$ and $Q$ be one-dimensional curves with $n$ vertices.
        For all $\alpha \in [1, n]$ and $\delta \geq 0$, there is an $\eps \leq \alpha \delta$ for which $P^\eps$ and $Q^\eps$ together have at most $n / \alpha$ edges of length at most $2\delta$. 
    \end{lemma}
    \begin{proof}
        Let $\alpha \in [1, n]$ and $\delta \geq 0$.
        If $P$ and $Q$ have at most $n/\alpha$ edges of length at most $2\delta$, then we may set $\eps \gets 0$, as $P^\eps = P$ and $Q^\eps = Q$ have the claimed number of short edges.
        Otherwise, we smooth the curves by a factor $\delta$.
        This gives the following recurrence on the amount of smoothing $\eps \coloneqq \eps(P, Q)$ required to obtain at most $2n/\alpha$ short edges:
        \[
            \eps(P, Q) = \begin{cases}
                0 & \parbox{20em}{if $P$ and $Q$ have at most $n/\alpha$ short edges,}\\
                \delta + \eps(P^\delta, Q^\delta) & \text{otherwise.}
            \end{cases}
        \]
        We argue that this recurrence stops after at most $\alpha$ iterations, and thus solves to $\eps \leq \alpha \delta$.
        
        Suppose $P$ and $Q$ have more than $n/\alpha$ edges of length at most $2\delta$.
        In the $\delta$-truncated smoothing of $P$ and $Q$, every short edge is truncated to a single point during truncation.
        Furthermore, while edges may ``merge'' during the procedure, they never ``split,'' so no new edges are created.
        Thus $P^\delta$ and $Q^\delta$ have at least $n/\alpha$ fewer short edges than $P$ and $Q$.
        It follows that within $\alpha-1$ iterations of the smoothing procedure, the number of short edges falls under the $n/\alpha$ threshold.
    \end{proof}

\subsection{Constructing truncated smoothings}
\label{sub:smoothing_algorithm}

    We present a linear time algorithm for constructing the truncated $\eps$-smoothing $P^\eps$ of $P$.
    It suffices to compute the points $p^\eps_1, \dots, p^\eps_n$ on $P^\eps$ that correspond to the vertices $p_1, \dots, p_n$ of~$P$.
    The curve $P^\eps$ can be obtained by interpolating linearly between the points $p^\eps_i$. 
    
    The algorithm relies on computing the \emph{death times} of the vertices $p_1, \dots, p_n$.
    We define the death time of $p_i$ to be the smallest value $\eps \geq 0$ for which $p_i$ is degenerate in $P^\eps$.
    Note that the death times of $p_1$ and $p_n$ are infinite, and that the death time of a degenerate vertex of $P$ is trivially $0$.
    In the following we show how to compute the death times of all non-degenerate vertices in $O(n)$ time in total.
    
    We express the death time of a non-degenerate vertex in terms of the extreme values in its sub- or superlevel set component.
    The \emph{sublevel set} of a point $p$ on $P$ is the set of points on $P$ with value at most $p$.
    The \emph{sublevel set component} of $p$ is the connected component of its sublevel set that contains $p$, see~\cref{fig:sublevel_curves_and_cartesian_tree}.
    The \emph{superlevel set component} of $p$ is defined symmetrically.
    For a local maximum $p_i$ of $P$, let $P^-$ be its sublevel set component.
    We define the points $\ell_i$ and $r_i$ as (global) minima on the prefix and suffix curves of $P^-$ that end and start at $p_i$, respectively.
    We let $m_i \coloneqq \min\{ |p_i - \ell_i|, |p_i - r_i| \}$, see~\cref{fig:sublevel_curves_and_cartesian_tree}.
    We symmetrically define $P^+$ to be the superlevel set component of a local minimum $p_i$ of $P$, and symmetrically define $\ell_i$ and $r_i$ in terms of $P^+$.
    The definition of $m_i$ is the same as for local maxima.
    For a degenerate vertex $p_i$ we set $m_i \coloneqq 0$.
    We show that the death time of an interior vertex $p_i$ is equal to $m_i / 2$.
    
    \begin{figure}
        \centering\includegraphics{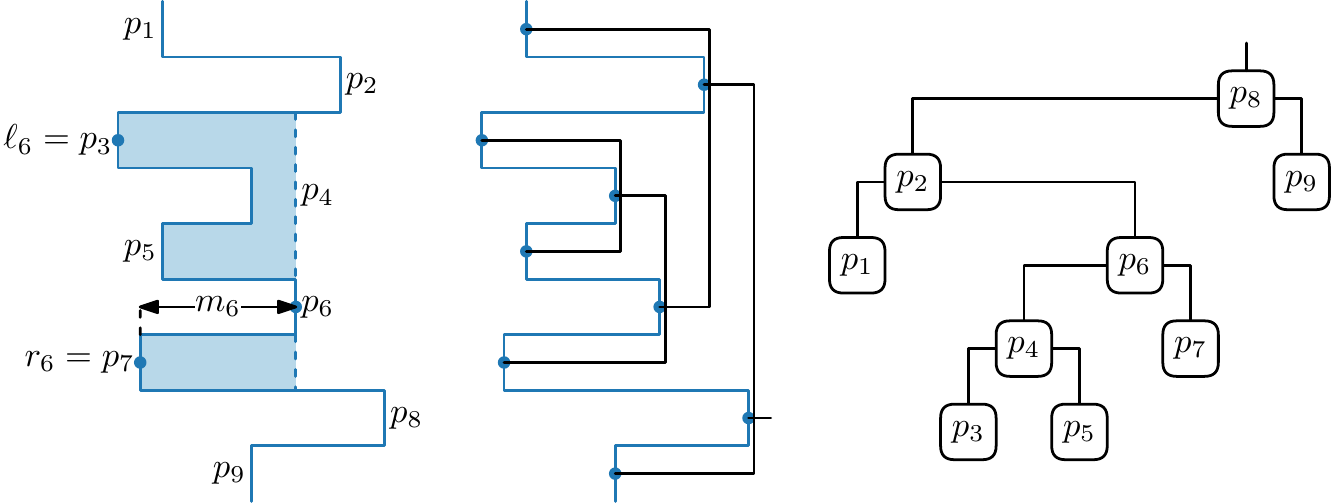}
        \caption{(left) The sublevel set component of $p_6$, left of the dashed line segment.
        Points $\ell_i$ and $r_i$ are the minima of the left and right parts of this component.
        (middle and right) The max-Cartesian tree built on the vertex sequence of~$P$.}
        \label{fig:sublevel_curves_and_cartesian_tree}
    \end{figure}
    
    \begin{lemma}
        For all $2 \leq i \leq n-1$ the death time of vertex $p_i$ is equal to $m_i / 2$.
    \end{lemma}
    \begin{proof}
        For any degenerate vertex $p_i$ of $P$ we have $m_i = 0$, which is trivially the death time of $p_i$.
        Let $p_i$ be a non-degenerate vertex of $P$, and assume without loss of generality that it is a local maximum.
        We distinguish between the case where $p_i$ is incident to a shortest edge of $P$ and the case where no incident edge has the minimum length.

        First assume that $p_i$ is incident to a shortest edge $e$ of $P$.
        Without loss of generality, we also assume that $e = \overline{p_{i-1} p_i}$.
        If $i = 2$ then the death time of $p_i$ is $\lVert e \rVert / 2$, as it coincides with $p_1$ in the truncated $\left( \lVert e \rVert / 2 \right)$-smoothing of $P$ and hence becomes degenerate.
        We have $\ell_2 = p_1$ and $r_2 \leq p_3 \leq p_1$, so $m_2 = |p_2 - p_1| = \lVert e \rVert$.
        Hence the death time of $p_2$ is $m_2 / 2$.

        If $i > 2$ then by the fact that $e$ has the minimum length of any edge we obtain that $p_{i-2} \geq p_i \geq p_{i-1}$ and $p_{i-1} \geq p_{i+1}$.
        The $\eps$-truncated smoothing of $P$, for $\eps = \lVert e \rVert / 2$, truncates $e$ to the point $p_i - \eps$, moves $p_{i-2}$ to $p_{i-2} - \eps \geq p_i - \eps$ and moves $p_{i+1}$ to $p_{i+1} + \eps \leq p_{i-1} + \eps = p_i - \eps$.
        This means that $p_i$ becomes degenerate, and thus its death time is $\eps = \lVert e \rVert / 2$.
        We have $\ell_i = p_{i-1}$ and $r_i \leq p_{i+1} \leq \ell_i$, hence $m_i = |p_i - \ell_i| = \lVert e \rVert$.
        See~\cref{fig:sublevel_curves_and_cartesian_tree}.
        Hence the death time of $p_i$ is $m_i / 2$.
    
        Next assume that $p_i$ is not incident to a shortest edge of $P$.
        Let $\eps$ be half the minimum edge length of $P$.
        Note that $\ell_i$ and $r_i$ are both local minima of $P$.
        As every local minimum of $P$ gets increased by $\eps$ by the simplification, every local maximum gets decreased by $\eps$, and the minimum edge length of $P$ is $2\eps$, we obtain that the points $\ell^\eps_i \coloneqq \ell_i + \eps$ and $r^\eps_i \coloneqq r_i + \eps$ are the analogues of $\ell_i$ and $r_i$ for the point $p^\eps_i$, with respect to $P^\eps$.
        It follows that $m^\eps_i$, the analogue of $m_i$, is equal to $\min\{ |p^\eps_i - \ell^\eps_i|, |p^\eps_i - r^\eps_i| \} = m_i - 2\eps$.
        Applying the above recursively on the point $p^\eps_i$, curve $P^\eps$ and value $m^\eps_i$ shows that the death time of $p_i$ is~$m_i / 2$.
    \end{proof}
    
    With the expression for the death times of interior vertices, we are able to compute the death times of these vertices in linear total time.
    To this end we use \emph{Cartesian trees}, introduced by Vuillemin~\cite{vuillemin80cartesian_tree}.
    A Cartesian tree is a binary tree with the heap property.
    We call a Cartesian tree a \emph{max-Cartesian tree} if it has the max-heap property and a \emph{min-Cartesian tree} if it has the min-heap property.
    A max-Cartesian tree $T$ for a sequence of values $x_1, \dots, x_n$ is recursively defined as follows.
    The root of $T$ contains the maximum value $x_j$ in the sequence.
    The subtree left of the root node is a max-Cartesian tree for the sequence $x_1, \dots, x_{j-1}$, and the right subtree is a max-Cartesian tree for the sequence $x_{j+1}, \dots, x_n$ (see~\cref{fig:sublevel_curves_and_cartesian_tree}).
    Min-Cartesian trees are defined symmetrically.
    
    \begin{lemma}
    \label{lem:computing_death_times}
        We can compute the death time of every interior vertex in $O(n)$ time altogether.
        The death times are reported as the sequence $m_2 / 2, \dots, m_{n-1} / 2$.
    \end{lemma}
    \begin{proof}
        To compute the death times we build two Cartesian trees; a max-Cartesian tree $T_{\max}$ and a min-Cartesian tree $T_{\min}$, both built on the sequence of vertices $p_1, \dots, p_n$ of $P$.
        These trees can be constructed in $O(n)$ time~\cite{gabow84scaling}.
        
        For a given node $\nu$ of $T_{\max}$ storing vertex $p_i$, the vertices stored in the subtree rooted at $\nu$ are precisely those in the sublevel set component of $p_i$.
        Thus if $p_i$ is a local maximum, the values $\ell_i$ and $r_i$ are precisely the minimum values stored in the left and right subtrees of $\nu$, respectively.
        We can therefore compute the death times of the interior local maxima of $P$ with a bottom-up traversal of $T_{\max}$, taking $O(n)$ time.
        Repeating the above process for $T_{\min}$, we compute the death times of the interior local minima of $P$ in $O(n)$ time as well.
        The death times of degenerate vertices are set to $0$, regardless of the values of $\ell_i$ and $r_i$.

        We augment the respective trees during their traversals to store the death times of vertices in the corresponding nodes.
        With in-order traversals of $T_{\max}$ and $T_{\min}$, interleaving the steps so we traverse the nodes in their order of index, we can construct the sequence of death times $m_2 / 2, \dots, m_{n-1} / 2$ in $O(n)$ time, rather than the $O(n \log n)$ time required to sort the death times by index.
    \end{proof}
    
    \begin{figure}
        \centering
        \includegraphics{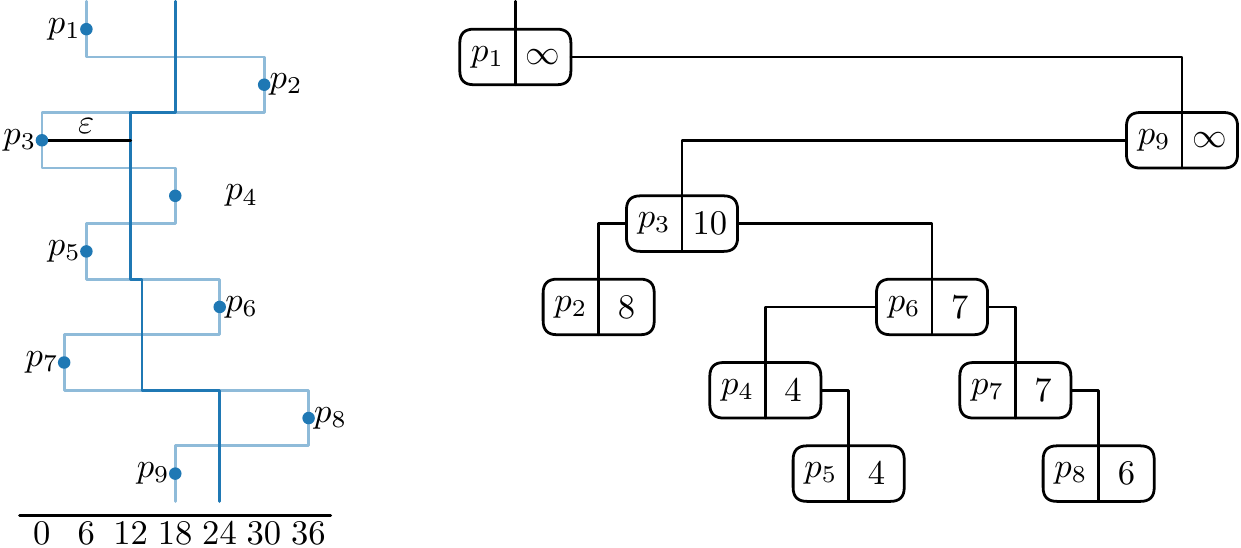}
        \caption{(left) The truncated $\eps$-smoothing for $\eps = 8$.
        The curve has three non-degenerate vertices, corresponding to the three vertices $p_1$, $p_3$ and $p_9$ with death times larger than $\eps$.
        (right) The death times of the vertices, stored in a max-Cartesian tree.}
        \label{fig:death_times}
    \end{figure}

    We construct the truncated $\eps$-smoothing in linear time using a max-Cartesian tree $T$ built on the sequence of death times reported by the algorithm of \cref{lem:computing_death_times} (see \cref{fig:death_times}).
    Since we do not compute $p^\eps_1, \dots, p^\eps_n$ in order of index, we augment $T$ to store these values in the corresponding nodes.
    Afterwards we extract $P^\eps$ from $T$ with an in-order traversal, without needing to sort the values based on index.
    This reduces the running time from $O(n \log n)$ to~$O(n)$.

    First we compute $p^\eps_1$ and $p^\eps_n$, as $p_1$ and $p_n$ behave differently from the rest during the truncated smoothing procedure.
    These vertices may change direction before their death time (which is infinite), whereas others move in a single direction until they become degenerate.
    To compute $p^\eps_1$ (and symmetrically $p^\eps_n$), we simulate the truncated smoothing procedure, but only for $p_1$.
    To this end we scan through the vertices of $P$.
    
    Let $p_i$ be the current vertex and let $\eps'$ be the current parameter.
    We keep track of the position of the first vertex $p^{\eps'}_1$ and maintain the invariant that $p^{\eps'}_i$ is the second non-degenerate vertex of $P^{\eps'}$ (the first being $p^{\eps'}_1$).
    Initially we set $p_i$ to the second non-degenerate vertex of $P$, and set $\eps' \coloneqq 0$ and $p^{\eps'}_1 \coloneqq p_1$.
    Let $\eps''$ be the death time of $p_i$.
    We distinguish between three cases to guide the process:

    \begin{enumerate}
        \item If $\eps'' = \infty$, so $p_i = p_n$, then either $p^\eps_1$ is the midpoint of the line segment $\overline{p_1^{\eps'} p_n^{\eps'}}$ (if $|p_1^{\eps'} - p_n^{\eps'}| \leq 2(\eps-\eps')$) or it is equal to the point on $\overline{p_1^{\eps'} p_n^{\eps'}}$ with distance $\eps-\eps'$ to $p_1^{\eps'}$.
        In this case we stop the scan, as we have computed $p_1^\eps$.

        \item If $\eps \leq \eps'' < \infty$, then $p^\eps_1$ is equal to the point on $\overline{p_1^{\eps'} p_i^{\eps'}}$ with distance $\eps-\eps'$ to $p_1^{\eps'}$.
        In this case, we stop the scan, as we have computed $p_1^\eps$.

        \item If $\eps > \eps''$, then $p^\eps_1$ is equal to the point on $\overline{p_1^{\eps'} p_i^{\eps'}}$ with distance $\eps''-\eps'$ to $p_1^{\eps'}$.
        We set $\eps' \gets \eps''$, and so $p^{\eps'}_1 \gets p^{\eps''}$.
        To set $p_i$ we continue the scan of the vertices until we have reached the first vertex with a death time greater than $\eps'$.
    \end{enumerate}

    With the above procedure we compute $p^\eps_1$ and $p^\eps_n$ in $O(n)$ time and store the values in the respective nodes of $T$.
    Afterwards we compute the locations of the interior vertices of $P^\eps$ by traversing $T$.
    For this we use the following technical lemma.

    \begin{lemma}
    \label{lem:vertex_location}
        Let $p_i$ be an interior vertex of $P$ with death time at most $\eps$.
        Let $p_\ell$ and $p_r$ be the last vertex before $p_i$, respectively the first vertex after $p_i$, with death times greater than that of $p_i$.
        Then $p^\eps_i$ is equal to the value on $\overline{p^\eps_\ell p^\eps_r}$ closest to either $p_i - m_i / 2$ (if $p_i$ is a local maximum) or $p_i + m_i / 2$ (if $p_i$ is a local minimum).
    \end{lemma}
    \begin{proof}
        Assume $p_i$ is a local maximum.
        The case when $p_i$ is a local minimum is symmetric.
        Let $\eps' = m_i / 2$ be the death time of $p_i$.
        The value of $p^{\eps'}_i$ is equal to $p_i - \eps'$.
        
        All points $p^{\eps'}_j$ with $\ell < j < r$ are degenerate, as their death times are at most $\eps'$.
        These points therefore do not move further during the truncated smoothing procedure, unless they coincide with either $p^{\eps''}_\ell$ or $p^{\eps''}_r$ for some $\eps' \leq \eps'' \leq \eps$.
        It follows that all these points lie on $\overline{p^\eps_\ell p^\eps_r}$.
        Specifically, $p^\eps_i$ is equal to $p^{\eps'}_i = p_i - m_i / 2$, unless $p^{\eps''}_i$ coincides with either $p^{\eps''}_\ell$ or $p^{\eps''}_r$ for some $\eps' \leq \eps'' \leq \eps$, in which case $p^\eps_i$ is equal to $p^\eps_\ell$, respectively $p^\eps_r$.
        
        Note that the orientation of $p^{\eps''}_\ell$ relative to $p^{\eps''}_r$ is the same for all $\eps' \leq \eps''$ (including the case where the points coincide).
        This is because the points are consecutive non-degenerate vertices for $\eps'' = \eps'$.
        To change orientation the points must attain the same value at some point, and since the subcurve between the points spans merely the line segment between the points, the points coincide at this point.
        Thus if $p^\eps_i$ is equal to either $p^\eps_\ell$ or $p^\eps_r$, then this endpoint is the point closest to~$p_i-m_i/2$.        
    \end{proof}

    With the above technical lemma we are ready to give the algorithm for computing the interior vertices $p^\eps_2, \dots, p^\eps_{n-1}$ of $P^\eps$.
    We perform a pre-order traversal of $T$, during which we compute the value of $p^\eps_i$ once we get to the node $\nu_i$ storing vertex $p_i$.
    During the traversal, we keep track of the nodes $\nu_\ell$ and $\nu_r$ that store the last vertex before $p_i$, respectively the first vertex after $p_i$.
    This is possible since $\nu_\ell$ and $\nu_r$ are both ancestors of $\nu_i$, so we have already visited them during the traversal.
    Since we have already visited them, we have also already computed the values of $p^\eps_\ell$ and $p^\eps_r$.
    We compute the value of $p^\eps_i$ in constant time using \cref{lem:vertex_location} and augment $\nu_i$ to store this value.
    Afterwards we extract $P^\eps$ from $T$ with a single in-order traversal.
    This gives the following result.

    \begin{lemma}
    \label{lem:constructing_smoothing}
        We can construct the truncated $\eps$-smoothing of a one-dimensional curve with $n$ vertices in $O(n)$ time, for any $\eps \geq 0$.
    \end{lemma}

\subsection{Finding the right parameter}

    Let $P$ and $Q$ be one-dimensional curves with $n$ vertices.
    Recall from~\cref{lem:smoothing_1D} that for every $\delta \geq 0$ and $\alpha \in [1, n]$, there exists a parameter $\eps \leq \alpha \delta$ such that $P^\eps$ and $Q^\eps$ together have at most $n / \alpha$ edges of length at most $2\delta$.
    In this section, we show how to compute a slightly worse parameter $\eps$, namely one where $P^\eps$ and $Q^\eps$ together have at most $2n/\alpha$ short edges, in $O(n)$ time.
    This parameter plays a crucial part in our approximate decision algorithms for the \f distance.
    
    We use the death times of the vertices of $P$ and $Q$ (see \cref{sub:smoothing_algorithm}).
    These can be computed in $O(n)$ total time.
    Let $M$ be the multiset of death times of the vertices of $P$ and $Q$.
    Any $\eps$ for which the half-open interval $(\eps, \eps + \delta]$ contains at most $n/\alpha$ elements of $M$ results in truncated smoothings $P^\eps$ and $Q^\eps$ with at most $2n/\alpha$ edges of length at most $2\delta$.
    This is because each edge of length at most $2\delta$ in $P^\eps$ and $Q^\eps$ is truncated to a point in the truncated $\delta$-smoothing of these curves and hence contains a vertex with death time at most $\delta$, which corresponds to a vertex of $P$ or $Q$ with death time at most $\eps + \delta$.
    Because a vertex with a death time in the range $(\eps, \eps + \delta]$ is incident to either one or two short edges, the number of short edges in the smoothing is at most $2n/\alpha$, rather than at most $n/\alpha$.

    It suffices to look for an $\eps \in M \cup \{0\}$, as any other value $\eps'$ with at most $n/\alpha$ death times in $(\eps', \eps' + \delta]$ implies that the highest value $\eps \in M \cup \{0\}$ below $\eps'$ has at most $n/\alpha$ death times in $(\eps, \eps + \delta]$. 
    The following lemma helps us in searching for a valid parameter $\eps$, as it allows us to discard half the death times of $M$ based on the median death time.

    \begin{lemma}
    \label{lem:median_death_time_property}
        Let $m^*$ be the median of $M$.
        If $m^* > \alpha \delta / 2$, then there exists a half-open interval $(\eps, \eps+\delta] \subseteq (0, \alpha \delta / 2]$ that contains at most $|M|/\alpha$ elements of $M$.
    \end{lemma}
    \begin{proof}
        Partition the interval $(0, \alpha \delta / 2]$ into $\alpha / 2$ disjoint half-open intervals of length $\delta$.
        For sake of contradiction say that each of these intervals has more than $|M| / \alpha$ elements of $M$ in it.
        Because the intervals are disjoint, there are more than $|M| / 2$ elements of $M$ in $(0, \alpha \delta / 2]$.
        The $(|M|/2)^{\mathrm{th}}$ death time, which is the median $m^*$, is therefore at most $\alpha \delta / 2$, giving a contradiction.
    \end{proof} 

    We give a recursive algorithm for computing a valid parameter ${\eps \in M \cup \{0\}}$.
    Specifically, the algorithm returns a value $\eps \coloneqq \eps(M, \alpha)$ such that there are at most $|M| / \alpha$ values of $M$ inside the interval $(\eps, \eps+\delta]$.
    We first compute $m^*$ in $O(|M|)$ time~\cite{blum73selection} and then proceed according to the following three cases.
    \begin{enumerate}
        \item If $|M| \leq n/\alpha$ then $\eps(M, \alpha) \coloneqq 0$ is a valid parameter.
        
        \item If $m^* > \alpha \delta / 2$ then by \cref{lem:median_death_time_property} there is a valid parameter $\eps \in [0, \alpha \delta / 2 - \delta]$.
        We limit our search to parameters in this range, which means that all death times greater than $\alpha \delta / 2$ can be discarded.
        We set $M' \coloneqq \{m \mid m \in M, m \leq \alpha \delta / 2\}$ and return $\eps(M', \alpha / 2)$.

        \item If $m^* \leq \alpha \delta / 2$ then we limit our search to parameters greater than $\alpha \delta / 2$.
        To be able to do this through recursion we discard death times that are at most $\alpha \delta / 2$ and modify the remaining death times to be $\alpha \delta / 2$ smaller.
        We set $M' = \{m - \alpha \delta / 2 \mid m \in M, m > \alpha \delta / 2\}$ and return $\eps(M', \alpha / 2) + \alpha \delta / 2$.
    \end{enumerate}
    
    The running time $T(M, \alpha)$ of the algorithm satisfies the recurrence
    \[
        T(M, \alpha) = \begin{cases}
            O(1) & \text{if $|M| \leq n/\alpha$,} \\
            T(M') + O(|M|) & \text{otherwise,}
        \end{cases}
    \]
    where $|M'| \leq |M| / 2$.
    This recurrence implies that $T(M, \alpha) = O(|M|)$ and hence the running time of the above algorithm is $O(n)$.
    The value of the returned parameter $\eps(M, \alpha)$ satisfies the recurrence
    \[
        \eps(M, \alpha) \leq \begin{cases}
            0 & \text{if $|M| \leq n/\alpha$,} \\
            \eps(M', \alpha / 2) + \alpha \delta / 2 & \text{otherwise,}
        \end{cases}
    \]
    where again $|M'| \leq |M| / 2$.
    This recurrence implies that $\eps \coloneqq \eps(M, \alpha) \leq \alpha \delta$.
    We obtain the following result.

    \begin{lemma}
    \label{lem:computing_smoothing_parameter}
        Let $P$ and $Q$ be one-dimensional curves with $n$ vertices.
        Let $\alpha \in [1, n]$ and $\delta \geq 0$.
        In $O(n)$ time, we can compute a parameter $\eps \leq \alpha \delta$ for which $P^\eps$ and $Q^\eps$ together have at most $2n / \alpha$ edges of length at most $2\delta$.
    \end{lemma}

    We present the main result of this section in terms of higher-dimensional curves.
    We refer to edges as monotone pieces again, to make clear that they can contain many vertices on them.
    
    \begin{theorem}
    \label{thm:reducing_short}
        Let $P$ and $Q$ be $d$-dimensional curves with $n$ vertices.
        Let $\alpha \in [1, n]$ and $\delta \geq 0$.
        In $O(n)$ time, we can construct curves $\hat{P}$ and $\hat{Q}$ whose projections onto the coordinate axes all have at most $2n/\alpha$ monotone pieces that are $2\delta$-short.
        Furthermore, we have $d_F(\hat{P}, \hat{Q}) \leq d_F(P, Q) \leq d_F(\hat{P}, \hat{Q}) + 2\alpha\delta$.
    \end{theorem}
    \begin{proof}
        Let $P_\ell$ and $Q_\ell$ be the projections of $P$ and $Q$ onto the $\ell^{\mathrm{th}}$ coordinate axis.
        For each dimension $\ell$, we compute a parameter $\eps_\ell \leq \alpha \delta$ such that $P_\ell^{\eps_\ell}$ and $Q_\ell^{\eps_\ell}$ together have at most $2n/\alpha$ monotone pieces that are $2\delta$-short.
        This takes $O(n)$ time~(\cref{lem:computing_smoothing_parameter}).
        We define the simplifications $\hat{P}$ and $\hat{Q}$ as:
        \[
            \hat{P}(x) = (P^{\eps_1}_1(x), \dots, P^{\eps_d}_d(x)) \qquad \text{and} \qquad \hat{Q}(y) = (Q^{\eps_1}_1(y), \dots, Q^{\eps_d}_d(y)).
        \]
        By constructing the individual truncated smoothings in $O(n)$ time each~(\cref{lem:constructing_smoothing}), we construct $\hat{P}$ and $\hat{Q}$ in $O(n)$ time.

        Let $\varphi = (f, g)$ be a matching between $P$ and $Q$ of minimum cost, so its cost is $\delta^* := d_F(P, Q)$.
        The reparameterizations $f$ and $g$ naturally induce a $\delta^*$-matching between $P_\ell$ and $Q_\ell$, for any dimension $\ell$.
        By~\cref{lem:smoothing_error_1D}, these reparameterizations also induce a $\delta^*$-matching between the truncated $\eps_\ell$-smoothings $P_\ell^{\eps_\ell}$ and $Q_\ell^{\eps_\ell}$.
        Due to the coordinate-wise definitions of $\hat{P}$ and $\hat{Q}$, and our use of the $L_\infty$ norm, this fact implies that $f$ and $g$ specify a matching between $\hat{P}$ and $\hat{Q}$ of cost at most $\delta^*$.

        Next let $\varphi = (f, g)$ be a matching between $\hat{P}$ and $\hat{Q}$ of minimum cost, so its cost is $\delta^* := d_F(\hat{P}, \hat{Q})$.
        The reparameterizations $f$ and $g$ naturally induce a $\delta^*$-matching between $P_\ell^{\eps_\ell}$ and $Q_\ell^{\eps_\ell}$, for any dimension $\ell$.
        By~\cref{lem:smoothing_error_1D}, these reparameterizations also induce a $(\delta^*+2\alpha\delta)$-matching between the projections $P_\ell$ and $Q_\ell$ of $P$ and $Q$.
        Due to the coordinate-wise definitions of $\hat{P}$ and $\hat{Q}$, and our use of the $L_\infty$ norm, this fact implies that $f$ and $g$ specify a matching between $P$ and $Q$ of cost at most $\delta^*+2\alpha\delta$.
    \end{proof}

\section{A faster approximation algorithm}
\label{sec:alternative_algorithm}

    With the simplification of~\cref{sec:reducing_narrow} at hand, we are ready to give our approximation algorithm for the \f distance, which significantly improves upon the state-of-the-art algorithm by Colombe and Fox~\cite{colombe21continuous_frechet}.
    Let $P$ and $Q$ be $d$-dimensional curves.
    We first consider the approximate decision problem with decision parameter $\delta \geq 0$ and approximation factor $O(\alpha)$.

    We apply the simplification of~\cref{thm:reducing_short} to compute a pair of simplifications $\hat{P}$ of $P$ and $\hat{Q}$ of $Q$, whose projections onto the coordinate axes all have at most $2n/\alpha$ $2\delta$-short monotone pieces.
    This makes the reachable $\delta$-free space constrained to lie in $O(n^2 / \alpha)$ blocks~(\cref{thm:narrow_vs_reachable}).
    The simplifications are constructed in $O(n)$ time and give an additive error to the \f distance of $2\alpha$~(\cref{thm:reducing_short}).

    Next we show that we can propagate reachability information through a block in time linear in its size, rather than quadratic.
    For this, we make use of the following observation.

    \begin{observation}
        Let $S$ be an ortho-convex set.
        There exists a bimonotone path between points $p, q \in S$ if and only if $p$ and $q$ lie in the same connected component of $S$.
    \end{observation}

    \begin{lemma}
    \label{lem:block_propagate_reachability}
        Given a block $B_{i, j}$ and all $\delta$-reachable points on the bottom and left sides of $B_{i, j}$, represented by a horizontal segment $e_B$ and vertical segment $e_L$, we can compute all $\delta$-reachable points on the top and right sides of $B_{i, j}$ in $O(|P_i| + |Q_j|)$ time.
    \end{lemma}
    \begin{proof}
        We first compute the $\delta$-close points on the top and right sides of $B_{i, j}$.
        These points form a horizontal segment $e_T$ and vertical segment $e_R$, and can be computed by constructing the free space inside the cells on the boundary of $B_{i, j}$.
        Since the free space inside a cell has constant complexity and can be constructed in constant time~\cite{alt95continuous_frechet}, computing these $\delta$-close points takes $O(|P_i| + |Q_j|)$ time.

        Next we determine which pairs of the four segments $e_B$, $e_L$, $e_T$ and $e_R$ lie in the same connected component of $\F(P_i, Q_j)$.
        For this, we explicitly compute the connected components of $\F(P_i, Q_j)$ that contain one of the line segments.
        By ortho-convexity of $\F(P_i, Q_j)$, its boundary $\partial \F(P_i, Q_j)$ can intersect only $O(|P_i| + |Q_j|)$ cells of the free space.
        Because the free space inside a cell has constant complexity and can be constructed in constant time, we may trace the boundary of a connected component of $\F(P_i, Q_j)$, starting at one of the four line segments, in $O(|P_i| + |Q_j|)$ time.
        All four segments lie on the boundary of $B_{i, j}$ and hence on the boundary of $\F(P_i, Q_j)$, so if a segment lies in the currently traced connected component, we encounter it while tracing the boundary.
        Thus after $O(|P_i| + |Q_j|)$ time, we determine which pairs of the four segments $e_B$, $e_L$, $e_T$ and $e_R$ lie in the same connected component of $\F(P_i, Q_j)$.

        The final step is to determine the maximum subedges of $e_T$ and $e_R$ containing points $(x', y')$ for which there is a point $(x, y) \in e_B \cup e_L$ with $x \leq x'$ and $y \leq y'$.
        This can be done in constant time.
    \end{proof}

    By~\cref{thm:narrow_vs_reachable} each monotone piece of $P$ and $Q$ corresponds to only $O(n / \alpha)$ blocks containing parts of the reachable free space.
    Using the algorithm of \cref{lem:block_propagate_reachability} to traverse these blocks, we get a free space traversal algorithm with running time
    \[
        O\Big( (n / \alpha) \cdot \Big[ \sum_i |P_i| + \sum_j |Q_j| \Big] \Big)
        = O(n^2 / \alpha).
    \]
    This is summarized in \cref{thm:summary_higher_dim}.

    \begin{theorem}\label{thm:summary_higher_dim}
        Let $P$ and $Q$ be $d$-dimensional curves with $n$ vertices.
        Given $\alpha \in [1, n]$ and $\delta \geq 0$, we can decide whether $d_F(P, Q) \leq (1+2\alpha) \delta$ or $d_F(P, Q) > \delta$ in $O(n^2 / \alpha)$ time.
    \end{theorem}

    To turn this decision algorithm into an approximation algorithm for the \f distance, we apply the black box technique of Colombe and Fox~\cite{colombe21continuous_frechet}.
    For any $\eps \in (0, 1]$, this increases the running time by a factor $\log (n / \eps)$ and the approximation factor by a factor $(1+\eps)$.
    We set $\eps = 1$ for concreteness, giving an $O((n^2 / \alpha) \log n)$-time $(2+4\alpha)$-approximation algorithm for the \f distance.
    To turn this algorithm into an $\alpha$-approximation algorithm running in the same time bound, we set $\alpha \gets (\alpha - 2) / 4$ for $\alpha \geq 6$, and run the exact quadratic-time decision algorithm of Alt and Godau~\cite{alt95continuous_frechet} for $\alpha \in [1, 6)$.
    This gives the following result:

    \begin{theorem}
        Let $P$ and $Q$ be $d$-dimensional curves with $n$ vertices.
        Given $\alpha \in [1, n]$, we can compute an $\alpha$-approximation to $d_F(P, Q)$ in $O((n^2 / \alpha) \log n)$ time.
    \end{theorem}

\section{A further improvement in one dimension}
\label{sec:faster_one_dimension}
    
    Finally, we present an improved algorithm for curves in one dimension.
    To this end, let $P$ and $Q$ be one-dimensional curves.
    We assume that the monotone pieces of $P$ and $Q$ are their edges, which can be accomplished by removing vertices of the curves while keeping the \f distance to the original curves at $0$.
    We again first consider the approximate decision problem with decision parameter $\delta \geq 0$ and approximation factor $O(\alpha)$, for $\alpha \in [1, n^{1/3}]$.
    
    Our approximate decision algorithm uses a subroutine for constructing approximate exit sets.
    The \emph{$\delta$-exit set} for a set of points $S \subseteq \{0\} \times [0, 1]$ with respect to $P$ and $Q$ is the set of all points $E(S) \subseteq \{1\} \times [0, 1]$ that are $\delta$-reachable from points in $S$.
    Points $(0, y) \in S$ represent matching $P(0)$ to $Q(y)$, and points $(1, y') \in E(S)$ represent matching $P(1)$ to $Q(y')$.
    For every point $(1, y') \in E(S)$, there exists a point $(0, y) \in S$ such that $P$ is within \f distance $\delta$ of the subcurve $Q[y, y']$.
    
    We allow for approximations, where instead of constructing an exact $\delta$-exit set for $S$, we construct an \emph{$(\alpha, \delta)$-exit set} $E_\alpha(S)$.
    Such a set must contain all of $E(S)$, and may only contain points that are $\alpha \delta$-reachable from points in $S$.
    We prove the following theorem in \cref{sec:exit_sets}:
    
    \begin{restatable}{theorem}{computingExitSets}
    \label{thm:computing_exit_sets}
        Let $P$ and $Q$ be one-dimensional curves with $n$ vertices.
        Let $S \subseteq \{0\} \times [0, 1]$ consist of $O(n)$ connected components.
        Given $\alpha \in [1, n^{1/2}]$ and $\delta \geq 0$, we can construct an $(\alpha, \delta)$-exit set $E_\alpha(S)$ for $S$ in $O((n^2 / \alpha^2) \log n)$ time.
    \end{restatable}

    We assume matchings are known to lie within $K = O(n/\alpha)$ blocks of the block diagonal.
    To accomplish this, we apply the  simplification of~\cref{thm:reducing_short} to reduce the number of $2\delta$-short monotone pieces of $P$ and $Q$ to at most $2n/\alpha$, which by~\cref{thm:narrow_vs_reachable_1D} makes the reachable $\delta$-free space constrained to within $4n/\alpha+1$ blocks of the block diagonal.
    Note that in this setting, blocks are cells of the parameter space.
    We therefore refer to cells instead of blocks in the remainder of this section.

    The total complexity of the reachable $\delta$-free space is only $O(n^2/\alpha)$, rather than $O(n^2)$.
    We wish to translate this lower complexity into lower input complexities for subroutines.
    For this we cover the reachable free space with $\alpha$ interior-disjoint rectangles, each $K \times 3K$ cells in dimension (except for three rectangles, which may be smaller), such that rectangles do not share a common $x$-coordinate on their interiors.
    That is, the rectangles are laid out from left to right over the reachable $\delta$-free space.
    See \cref{fig:1D_algorithm} for an illustration.
    \begin{figure}
        \centering
        \includegraphics{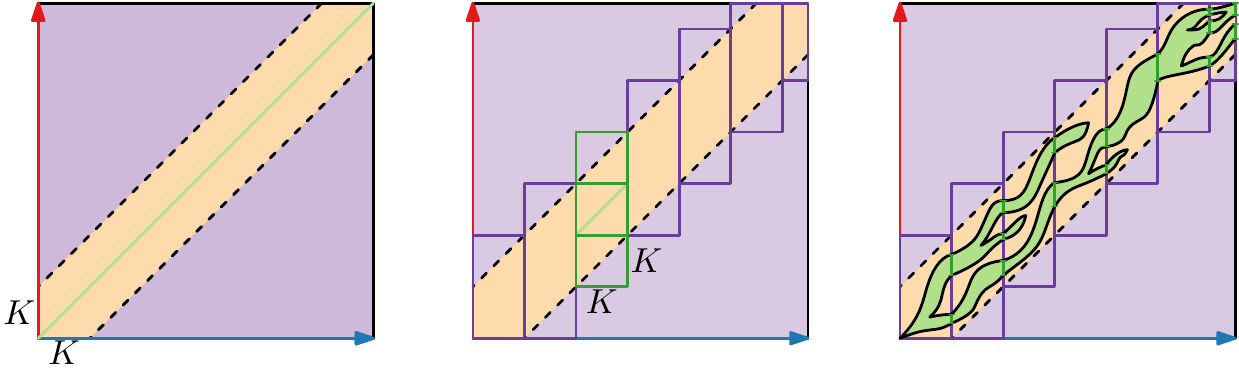}
        \caption{(left) Matchings are constrained to lie within $K$ cells of the cell diagonal.
        (middle) The $\alpha$ rectangles covering all possible matchings.
        Each rectangle spans $K$ cells in width and $3K$ cells in height.
        (right) Constructing exit sets in each rectangle, given the exit set of the previous rectangle.}
        \label{fig:1D_algorithm}
    \end{figure}
    
    We iteratively process the rectangles from left to right, constructing exit sets inside each rectangle for given sets on their left boundary.
    Let $R_1, \dots, R_\alpha$ be the rectangles in left to right order.
    Let $P_1, \dots, P_\alpha$ be the subcurves of $P$ corresponding to these rectangles, and let $Q_1, \dots, Q_\alpha$ be the subcurves of $Q$ corresponding to these rectangles.
    For each rectangle $R_i$ we construct an $(\alpha, \delta)$-exit set $E_i = E_\alpha(S_i)$ for the set $S_i = E_{i-1} \cap \F(P, Q)$, with respect to $P_i$ and $Q_i$.
    We construct these exit sets using the algorithm of \cref{sec:exit_sets}.
    
    Initially, $S_1 = \{(0, 0)\}$.
    Given a set $S_i$, we construct $E_i$ using \cref{thm:computing_exit_sets} to construct an $(\alpha, \delta)$-exit set $E_\alpha(S_i)$ with respect to $P_i$ and $Q_i$.
    Because $P_i$ and $Q_i$ have only $O(K)$ vertices each, we construct $E_i$ in $O((K^2 / \alpha^2) \log K)$ time.
    We then construct the set $S_{i+1}$ as the intersection between $E_i$ and $\F(P_i, Q_i)$, which takes $O(K \log K)$ time by sorting the sets and scanning over them.
    
    Performing the above for all $\alpha$ rectangles, with $K = O(n / \alpha)$, gives an $\alpha$-approximate decision algorithm with running time $O((n^2 / \alpha^3) \log n)$, assuming all $\delta$-matchings lie within $K$ cells of the cell diagonal.
    This assumption will be justified after applying the truncated smoothings of \cref{sec:reducing_narrow}, which may incur an error of $2\alpha$, and overall results in a $3\alpha$-approximate decider.
    
    \begin{theorem}
        Let $P$ and $Q$ be one-dimensional curves with $n$ vertices.
        Given $\alpha \in [1, n^{1/3}]$ and $\delta \geq 0$, we can decide whether $d_F(P, Q) \leq 3\alpha \delta$ or $d_F(P, Q) > \delta$ in $O((n^2 / \alpha^3) \log n)$ time.
    \end{theorem}

    To turn this decision algorithm into an approximation algorithm for the \f distance, we apply the black box technique of Colombe and Fox~\cite{colombe21continuous_frechet}.
    For any $\eps \in (0, 1]$, this increases the running time by a factor $\log (n / \eps)$ and the approximation factor by a factor~$(1+\eps)$.
    Concretely, we set $\eps = 1$, which gives an $O((n^2 / \alpha^3) \log^2 n)$-time $6\alpha$-approximation algorithm for the \f distance.
    To turn this algorithm into an $\alpha$-approximation algorithm running in the same time bound, we set $\alpha \gets \alpha / 6$ for $\alpha \geq 6$, and run the exact quadratic-time decision algorithm of Alt and Godau~\cite{alt95continuous_frechet} for $\alpha \in [1, 6)$.
    This gives the following result:

    \begin{theorem}
        Let $P$ and $Q$ be one-dimensional curves with $n$ vertices.
        Given $\alpha \in [1, n^{1/3}]$, we can compute an $\alpha$-approximation to $d_F(P, Q)$ in $O((n^2 / \alpha^3) \log^2 n)$ time.
    \end{theorem}

\section{Constructing approximate exit sets}
\label{sec:exit_sets}
    
    In this section we present an algorithm for efficiently constructing $(O(\alpha), \delta)$-exit sets for one-dimensional curves $P$ and $Q$ with $n$ vertices, for $\alpha \in [1, n^{1/3}]$.
    In \cref{sub:exit_set_segment} we first present a data structure for constructing approximate exit sets for sets of the form $\{0\} \times [y_1, y_2]$ when $P$ is a line segment.
    Then in \cref{sub:approx_exit_interval} we extend this data structure to construct exit sets for when $P$ is a ``good'' curve.
    Finally, in \cref{sub:exit_set_general_case} we use this data structure to construct exit sets for the general case, where we construct exit sets for general sets of points, making no assumptions on $P$.

\subsection{Exit sets for single line segments}
\label{sub:exit_set_segment}

    First we give an algorithm for constructing approximate exit sets for sets of the form $\{0\} \times [y_1, y_2]$, with respect to a directed line segment $e$ and the curve $Q$.
    The construction takes $O(\log n)$ time, after preprocessing $Q$ in $O(n \log n)$ time.
    The quality of the resulting approximate exit set depends on the diameter of $Q[y_1, y_2]$.
    Specifically, let its diameter be~$D\delta$.
    Then the returned set is a $(D+3, \delta)$-exit set.

    Let $Q(y'_1)$ be the first point on $Q[y_1, 1]$ within distance $\delta$ of $e(1)$, the last endpoint of $e$.
    If this point does not exist, then we report the empty set, as it is a $(1, \delta)$-exit set.
    Otherwise, we have that either ($i$) $y'_1 \in [y_1, y_2]$ or ($ii$) $y'_1 \in (y_2, 1]$.
    
    In case ($i$), it follows from the triangle inequality that all of $Q[y_1, y_2]$ lies within distance $(D+1)\delta$ of $e(1)$.
    Furthermore, if there exists a point $Q(y)$ with $y \in [y_1, y_2]$ that is within distance $\delta$ of $e(0)$, then the length of $e$ is at most $|Q(y) - Q(y'_1)| + 2\delta \leq (D+3)\delta$.

    Any curve within \f distance $\delta$ of $e$ is contained in the interval of points with distance at most $\delta$ to $e$.
    Thus, they are contained within distance $\lVert e \rVert + \delta$ of the last endpoint $e(1)$ of $e$.
    In particular, any subcurve $Q[y, y']$ with $y \in [y_1, y_2]$ and $d_F(e, Q[y, y']) \leq \delta$ lies within distance $(D+4)\delta$ of $e(1)$.
    The maximum interval $[y, y']$ for which $Q[y, y']$ has this property corresponds to a $(D+4, \delta)$-exit set for points in $\{0\} \times [y_1, y_2]$.
    See \cref{fig:exit_set_segment} (left).
    \begin{figure}
        \centering
        \includegraphics[page=1]{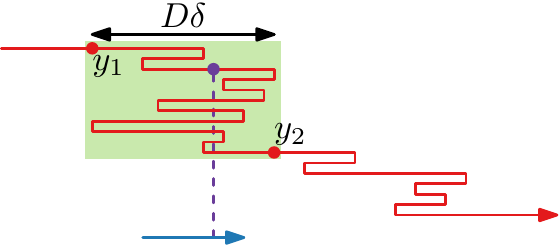}
        \quad
        \includegraphics[page=2]{exit_set_segment.pdf}
        \caption{The two types of constructed exit sets (green).
        The purple points are $Q(y'_1)$.
        Left is the case where $y'_1 \in [y_1, y_2]$, right is the case where $y'_1 \in (y_2, 1]$.}
        \label{fig:exit_set_segment}
    \end{figure}
    We report this interval, and with that a $(D+4, \delta)$-exit set in case ($i$), in $O(\log n)$ time using the data structure of \cref{lem:ds_maximal_interval}:

    \begin{lemma}
    \label{lem:ds_maximal_interval}
        We can preprocess $Q$ in $O(n \log n)$ time, such that given a value $\eps \geq 0$, a point $p \in \R$, and a point $q$ on $Q$, we can report the first point $Q(y_1)$ after $q$ that is within distance $\eps$ of $p$.
        Such a query takes $O(\log n)$ time.
        Additionally, we can report the last point $Q(y_2)$ for which $Q[y_1, y_2]$ lies completely within distance $\eps$ of $p$ in the same asymptotic time bound.
    \end{lemma}
    \begin{proof}
        We make use of a data structure for \emph{orthogonal range successor} queries.
        Given a set $S$ of points in the plane, together with an axis-aligned rectangle $R$, an orthogonal range successor query asks for a point in $S \cap R$ with minimum $y$-coordinate.
        Such queries can be answered in $O(\log |S|)$ time, after $O(|S| \log |S|)$ preprocessing time, by storing $S$ in a range tree with fractional cascading~\cite{chazelle86fractional_cascading,agarwal04range_searching}.

        We use orthogonal range successor queries in the following manner.
        Let $q_1, \dots, q_m$ be the $m = O(n)$ vertices of $Q$.
        We built the data structure on $S = \{(q_j, j) \mid 1 \leq j \leq m\}$.
        
        Let $\eps \geq 0$, a point $p \in \R$, and a point $q$ on $Q$ be given.
        We report the first point $Q(y_1)$ after $q$ that is within distance $\eps$ of $p$ by first checking if $|p-q| \leq \eps$, reporting $q$ as the answer if so.
        Otherwise, we check if $q < p-\eps$ or $q > p+\eps$.
        In the remainder of the query algorithm we assume $q < p-\eps$; the other case is symmetric.

        Let $Q(y) = q$.
        We determine the first vertex $q_j$ after $q$ that lies to the right of $p-\eps$.
        This vertex can be computed with an orthogonal range successor query on $S$ with the range $R = [p-\eps, \infty) \times [y, \infty)$.
        The edge $\overline{q_{j-1} q_j}$ contains the point $Q(y_1)$.
        In fact, $Q(y_1)$ is exactly the point on this edge with value $p-\eps$.
        We report this point, answering the first type of query after $O(\log n)$ time.

        For the second type of query, where we wish to report the last point $Q(y_2)$ for which $Q[y_1, y_2]$ lies completely within distance $\eps$ of $p$, we use a similar algorithm.
        We determine the first vertex $q_{j'}$ after $Q(y_1)$ that has distance greater than $\eps$ to $p$, i.e., $q_{j'} < p-\eps$ or $q_{j'} > p+\eps$.
        We do so with two orthogonal range successor queries.
        Afterwards, we report $Q(y_2)$ as the point on $\overline{q_{j'} q_{j'-1}}$ with distance exactly $\eps$ to $p$.
        This answers the second type of query after $O(\log n)$ time.
    \end{proof}

    For case ($ii$), consider a subcurve $Q[y, y']$ with $y \in [y_1, y_2]$ and $d_F(e, Q[y, y']) \leq \delta$ (if one exists) and let $\varphi$ be an arbitrary $\delta$-matching between the curves.
    Gudmundsson~\etal~\cite{gudmundsson19long} observe that $Q[y'_1, y']$ must lie completely within distance $3\delta$ of $e(1)$.\footnote{
        Gudmundsson~\etal show the propery for segments of length greater than $2\delta$, but the statement naturally holds for shorter segments as well.
    }
    Hence the maximal interval $[y'_1, y'_2]$ with $Q[y'_1, y'_2]$ completely within distance $3\delta$ of $e(1)$ contains all possible values $y'$ for which $d_F(e, Q[y, y']) \leq \delta$.
    We report this interval with the data structure of \cref{lem:ds_maximal_interval}.

    In case ($ii$), the interval $[y'_1, y'_2]$ corresponds to an $(D+3)$-approximate exit set for points in $\{0\} \times [y_1, y_2]$, if a non-empty exit set exists.
    In order for a non-empty exit set to exist, it must be the case that $d_F(e, Q[y_1, y'_2]) \leq (D+3)\delta$.
    We use the following data structure by Blank and Driemel~\cite{blank24imbalanced_Frechet} to decide whether this is the case.

    \begin{lemma}[{\cite[Theorem~13]{blank24imbalanced_Frechet}}]
        We can preprocess $Q$ in $O(n \log n)$ time, such that given an interval $[y_1, y_2] \subseteq [0, 1]$ and a directed line segment $e$, we can decide whether $d_F(e, Q[y_1, y_2])$ in $O(\log n)$ time.
    \end{lemma}

    After querying the above data structure to determine whether $d_F(e, Q[y_1, y'_2]) \leq (D+3)\Delta$, we either report $[y'_1, y'_2]$ as a $(D+3, \delta)$-exit set, or report the empty exit set.
    The total time taken to construct $[y'_1, y'_2]$, and subsequently decide whether to report it, is $O(\log n)$.

    In case ($ii$), the reported $(D+3, \delta)$-exit set is naturally also a $(D+4, \delta)$-exit set.
    Thus, in both cases, we report a $(D+4, \delta)$-exit set.
    The time taken is $O(\log n)$, after $O(n \log n)$ preprocessing time.
    We thus obtain:
    
    \begin{lemma}
    \label{lem:ds_exit_set_segment}
        We can preprocess $Q$ in $O(n \log n)$ time, such that given a directed line segment $e$ and an interval $[y_1, y_2]$, we can report a $(D+4, \delta)$-exit set for $\{0\} \times [y_1, y_2]$ in $O(\log n)$ time, where $D\delta$ is the diameter of $Q[y_1, y_2]$.
        The reported exit set is of the form $\{1\} \times [y'_1, y'_2]$ with $Q[y'_1, y'_2]$ completely within distance $3\delta$ of $e(1)$.
    \end{lemma}

\subsection{Exit sets for interior-good subcurves}
\label{sub:approx_exit_interval}

    In this section we extend the data structure of \cref{sub:exit_set_segment} to efficiently construct exit sets for singleton sets $\{(0, y)\}$ with respect to ``interior-good'' subcurves of $P$, rather than mere line segments.
    We define what good subcurves are in \cref{subsub:good_curves}, showing a useful property regarding the \f distance.
    Then we use this property in \cref{subsub:interior_good_curve} to construct exit sets for interior-good curves.
    
\subsubsection{Good curves}
\label{subsub:good_curves}
    
    First we make the observation that even though we are working in the continuous setting, the possible matchings are all relatively discrete.
    For this we make use of \emph{$\delta$-signatures}, introduced by Driemel~\etal~\cite{driemel15clustering}.
    The $\delta$-signature $\Sigma_\delta(P)$ of a one-dimensional curve $P$ is the curve defined by the series of values $0 = s_1 < \dots < s_k = 1$ with the following properties:
    \begin{enumerate}
        \item (non-degeneracy)~
        For all $2 \leq i \leq k-1$:\\
        $P(s_i) \notin \overline{P(s_{i-1}) P(s_{i+1})}$.
        \item (direction-preserving)~
        For all $1\leq i\leq k-1$:\\
        If $P(s_i) < P(s_{i+1})$, then for all $s, s' \in [s_i, s_{i+1}]$ with $s < s'$, $P(s) - P(s') \leq 2\delta$.\\
        If $P(s_i) > P(s_{i+1})$, then for all $s, s' \in [s_i, s_{i+1}]$ with $s < s'$, $P(s') - P(s) \leq 2\delta$.
        \item (minimum edge length)~
        If $k > 2$, then for all $2 \leq i \leq k-2$:\\
        $|P(s_{i+1}) - P(s_i)| > 2\delta$.\\
        $|P(s_{2}) - P(s_1)| > \delta$ and $|P(s_k)-P(s_{k-1})|>\delta$.
        \item (range)~
        For all $s \in [s_i, s_{i+1}]$ with $1 \leq i \leq k-1$:\\
        If $2 \leq i \leq k-2$, then $P(s) \in \overline{P(s_i) P(s_{i+1})}$.\\
        If $i = 1$ and $k > 2$, then $P(s) \in \overline{P(s_1) P(s_{2})} \cup [P(s_1) - \delta, P(s_1) + \delta]$.\\
        If $i = k-1$ and $k > 2$, then $P(s) \in \overline{P(s_{k-2} P(s_{k-1})} \cup [P(s_{k-1}) - \delta, P(s_{k-1}) + \delta]$.\\
        If $i = 1$ and $k = 2$, then $P(s) \in \overline{P(s_1) P(s_2)} \cup [P(s_1) - \delta, P(s_1) + \delta] \cup [P(s_2) - \delta, P(s_2) + \delta]$.
    \end{enumerate}
    The first and fourth properties imply that the parameters $s_i$ specify vertices of $P$.
    We refer to these vertices as the \emph{$\delta$-signature vertices} of $P$.
    
    \begin{lemma}
    \label{lem:signature_matchings}
        Any $\delta$-matching between one-dimensional curves $P$ and $Q$ matches $\delta$-signature vertices of $P$ to points on edges of $Q$ with an endpoint within distance $\delta$.
    \end{lemma}
    \begin{proof}
        Let $\sigma_1, \dots, \sigma_k$ be the $\delta$-signature vertices of $P$.
        Fix a $\delta$-matching $\varphi$ between $P$ and $Q$.
        By definition of the \f distance, $\varphi$ matches $\sigma_1$ to $Q(0)$ and $\sigma_k$ to $Q(1)$, making the statement hold immediately for these signature vertices.
        Next consider a signature vertex $\sigma_i$, for some $1 < i < k$.
        Assume that $\sigma_{i-1} < \sigma_i > \sigma_{i+1}$ (the other case is symmetric and hence omitted).
        Let $Q(x_{i-1}),~Q(x_i)$ and $Q(x_{i+1})$ be the points on $Q$ matched to $\sigma_{i-1},~\sigma_i$ and $\sigma_{i+1}$, respectively.
        It follows from the range property that $Q[x_{i-1}, x_{i+1}]$ lies completely left of $\sigma_i + \delta$.
        Hence the edge of $Q$ containing $Q(x_i)$, which is also an edge of $Q[x_{i-1}, x_{i+1}]$, must have an endpoint in $[Q(x_i), \sigma_i + \delta] \subseteq [\sigma_i - \delta, \sigma_i + \delta]$.
    \end{proof}

    \begin{remark*}
        Blank and Driemel~\cite{blank24imbalanced_Frechet} show that if $d_F(P, Q) \leq \delta$, then there exists a $\delta$-matching that matches the $\delta$-signature vertices of the curves to vertices of the other curve, thus making the signature vertices behave ``fully'' discrete.
        To construct exit sets however, we require knowledge about all $\delta$-matchings, and so we make use of~\cref{lem:signature_matchings} instead.
    \end{remark*}
    
    With the above observation, continuous matchings can be seen as (almost) discrete when restricted to signature vertices of one of the curves.
    We use this to apply ideas of Chan and Rahmati~\cite{chan18improved_approximation}, who give an approximation algorithm for the discrete \f distance.
    
    Their main idea is to identify bad vertices.
    Consider an infinite grid $\G$ with cellwidth~$\alpha \delta$.
    A point $p$ is \emph{$c$-bad}, for some $c \geq 0$, if it is within $c \delta$ of the boundary of $\G$.
    A point is \emph{$c$-good} if it is not $c$-bad.
    See \cref{fig:example_grid}.
    \begin{figure}
        \centering
        \includegraphics{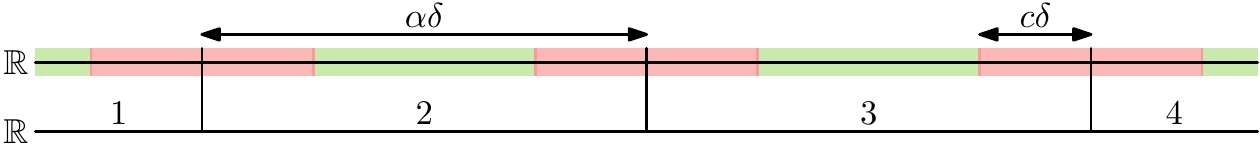}
        \caption{The real line $\R$ partitioned into grid cells of size $\alpha \delta$.
        On the top line the $c$-bad points are shown in red and the $c$-good points in green.
        On the bottom line the grid cells are labelled.}
        \label{fig:example_grid}
    \end{figure}
    A curve $P$ is $c$-good if all its vertices are $c$-good.

    Curves with good signatures have a particularly useful property, linking the \f distance to exact string matching.
    We assume that no vertex of either curve lies on the boundary of $\G$, which can be achieved by an infinitesimal shift of the curves without affecting the distances between points.
    We assign unique integer labels to the grid cells, through some order-preserving bijection from $\Z$, and let $\ell(v)$ denote the label of the grid cell containing a vertex $v$.
    The \emph{label curve} of a curve $P$ with vertices $p_1, \dots, p_n$ is the continuous curve with vertices $\ell(p_1), \dots, \ell(p_n)$.
    The main property linking the \f distance to exact string matching is given in \cref{lem:equal_label_curves}.
    To prove this lemma, we use the concept of \emph{visiting orders}, introduced by Bringmann~\etal~\cite{bringmann22ann_frechet}.
    A \emph{$\delta$-visiting order} of a curve $P$ with vertices $p_1, \dots, p_n$ on some curve $Q$ with vertices $q_1, \dots, q_m$ is a sequence of indices $j_1 \leq \dots \leq j_n$ such that $|p_i - q_{j_i}| \leq \delta$ for all $i$.
    We use the following results on signatures and visiting orders (rephrased as done in~\cite{bringmann22ann_frechet}):

    \begin{lemma}[{\cite[Lemma~3.1]{driemel15clustering}}]
    \label{lem:distance_to_signature}
        For any $\delta \geq 0$ and curve $P$ we have $d_F(P, \Sigma_\delta(P)) \leq \delta$.
    \end{lemma}
    
    \begin{lemma}[{\cite[Lemma~3.2]{driemel15clustering}}]
    \label{lem:visiting_order}
        For any $\delta \geq 0$ and curves $P$ and $Q$ with $d_F(P, Q) \leq \delta$, there exists a $\delta$-visiting order of $\Sigma_\delta(P)$ on $Q$.
    \end{lemma}
    
    \begin{lemma}
    \label{lem:equal_label_curves}
        Let $P$ be a curve whose $\delta$-signature is $5$-good and let $Q$ be an arbitrary curve.
        If $d_F(P, Q) \leq \delta$ then the label curves of $\Sigma_\delta(P)$ and the $\Sigma_{2\delta}(Q)$ are equal (up to reparameterization).
        Conversely, if the label curves of $\Sigma_\delta(P)$ and the $\Sigma_{2\delta}(Q)$ are equal, then $d_F(P, Q) \leq (\alpha+3) \delta$.
    \end{lemma}
    \begin{proof}
        Let $\ell_1, \dots, \ell_k$ be the extrema of the label curve of $\Sigma_\delta(P)$.
        These labels correspond to subcurves $\Sigma_1, \dots, \Sigma_k$ inside the corresponding grid cells.
        Note that these subcurves do not combine to make $\Sigma_\delta(P)$, since we ignore labels that are not extreme.
        We make use of the following claim:
        \begin{claim}
        \label{clm:extremal_labels}
            Every $\Sigma_i$ contains a $3\delta$-signature vertex of $P$.
        \end{claim}
        \begin{claimproof}
            The case where $i \in \{1, k\}$ holds trivially.
            Hence assume that $1 < i < k$.
            Further assume without loss of generality that $\ell_{i-1}, \ell_{i+1} < \ell_i$.
            The curve $\Sigma_i$ starts and ends on the left boundary of the cell with label $\ell_i$, staying inside this cell.
            Let $\sigma$ be a rightmost vertex of $\Sigma_i$.
            We assume for ease of exposition that this vertex is unique.
            By definition $\sigma$ is a $\delta$-signature vertex.
            We show that $\sigma$ is also a $3\delta$-signature vertex.
        
            Suppose for a contradiction that $\sigma$ is not a $3\delta$-signature vertex.
            By the range property, there are $3\delta$-signature vertices $\sigma_j, \sigma_{j+1}$ of $P$ with $\sigma$ on the subcurve between them and $\sigma \in \overline{\sigma_j \sigma_{j+1}}$.
            Since we assumed the rightmost vertex of $\Sigma_i$ to be unique, both $\sigma_j$ and $\sigma_{j+1}$ cannot be interior to $\Sigma_i$, as otherwise the range property would be violated by one of the endpoints of $\Sigma_i$.
            However, this violates the direction-preserving property.
            This is because the subcurve between $\sigma_j$ and $\sigma_{j+1}$ must have edges crossing the left boundary of cell $\ell_i$, once from left to right and once from right to left.
            These edges must both have a length greater than $10\delta$, as they cross an interval of length $10\delta$ around the boundary where no vertices of $\Sigma_\delta(P)$ lie.
            Therefore $\sigma$ must be a $3\delta$-signature vertex of~$P$.
        \end{claimproof}
    
        By \cref{lem:distance_to_signature,lem:visiting_order} there is a $3\delta$-visiting order of $\Sigma_{3\delta}(P)$ on $\Sigma_{2\delta}(Q)$.
        By \cref{clm:extremal_labels} every $\Sigma_i$ contains a $3\delta$-signature vertex of $P$.
        Recall that these vertices are $5$-good.
        Hence the $3\delta$-visiting order implies that there is a sequence of $2\delta$-signature vertices $\tau_1, \dots, \tau_k$ of $Q$ with each $\tau_i$ lying in the same grid cell as $\Sigma_i$.
        All of these vertices are $2$-good, because every $\delta$-signature vertex of $P$ is $5$-good.
    
        By \cref{lem:distance_to_signature,lem:visiting_order} there is also a $2\delta$-visiting order of $\Sigma_{2\delta}(Q)$ on $\Sigma_\delta(P)$.
        This, together with the fact that each $\tau_i$ is $2$-good, implies that there is a sequence of $\delta$-signature vertices $\sigma'_1, \dots, \sigma'_{k'}$ of $P$ with each $\sigma'_i$ lying inside the same grid cell as $\tau_i$.
        The extrema defined by these grid cells are equal to the extrema of the label curve of $\Sigma_\delta(P)$, meaning the label curves of $\Sigma_\delta(P)$ and the $\Sigma_{2\delta}(Q)$ are the same (up to reparameterization).

        Next suppose that the label curves of $\Sigma_\delta(P)$ and $\Sigma_{2\delta}(Q)$ are equal.
        We show that $d_F(P, Q) \leq (\alpha+3)\delta$.
        For this, we construct an $\alpha\delta$-matching between $\Sigma_\delta(P)$ and $\Sigma_{2\delta}(Q)$ by matching points on $\Sigma_\delta(P)$ to points on $\Sigma_{2\delta}(Q)$ inside the same grid cell.
        This matching naturally has cost at most $\alpha \delta$, and by the triangle inequality and \cref{lem:distance_to_signature} we get $d_F(P, Q) \leq d_F(P, \Sigma_\delta(P)) + d_F(\Sigma_\delta(P), \Sigma_{2\delta}(Q)) + d_F(\Sigma_{2\delta}(Q), Q) \leq (\alpha+3)\delta$.
    \end{proof}

\subsubsection{Exit sets for interior good subcurves}
\label{subsub:interior_good_curve}

    Next we give a data structure for constructing approximate exit sets for interior-good subcurves of $P$.
    We call a subcurve $P'$ of $P$ \emph{interior $c$-good} if all interior vertices of its $\delta$-signature are $c$-good.

    Let $P' = P[s_{i-1}, s_{j+1}]$ be an interior $6$-good subcurve.
    Given a point $z = (0, y) \in \F(P', Q)$, we show through a construction that there exists a $(\alpha + 7, \delta)$-exit set for $\{z\}$ consisting of a single connected component.
    This construction can be performed in $O(\log n)$ time, after having preprocessed $P$ and $Q$ into data structures in $O(n \log n)$ time.

    Let $P(s_i)$ and $P(s_j)$ be the first, respectively last, $\delta$-signature vertices on the interior of $P'$.
    Let $s'_i$ and $s'_j$ index these vertices on $P'$, so $P'(s'_i)$ is the vertex $P(s_i)$ and $P'(s'_j)$ is the vertex $P(s_j)$.
    We first construct an approximate exit set for $\{z\}$ with respect to $P'[0, s'_i]$ and $Q$.
    For this, note that because $P'[0, s'_i] = P[s_{i-1}, s_i]$ is the subcurve between two consecutive signature vertices, we have $d_F(P'[0, s'_i], \overline{P'(0) P'(s'_i)}) \leq \delta$ by \cref{lem:distance_to_signature}.
    We construct a $(4, 2\delta)$-exit set for $\{z\}$ with respect to $\overline{P'(0) P'(s'_i)}$ and $Q$ with the data structure of \cref{lem:ds_exit_set_segment}.
    The returned exit set is of the form $\{1\} \times [y_1, y_2]$, where $Q[y_1, y_2]$ is completely within distance $6\delta$ of $P'(s'_i)$.

    Let $\{1\} \times [y_1, y_2]$ be the returned exit set.
    Observe that it is a $(9, \delta)$-exit set with respect to $P'[0, s'_i]$ and $Q$.
    Indeed, any $\delta$-matching between $P'[0, s'_i]$ and a subcurve $Q'$ of $Q$ is also a $2\delta$-matching between $\overline{P'(0) P'(s'_i)}$ and $Q'$, so the exit set contains all $\delta$-reachable points.
    Conversely, any $8\delta$-matching between $\overline{P'(0) P'(s'_i)}$ and $Q'$ is also a $9\delta$-matching between $P'[0, s'_i]$ and $Q'$, so the exit set contains only $9\delta$-reachable points.

    Next we extend the exit set $\{1\} \times [y_1, y_2]$ into one with respect to $P'[0, s'_j]$ and $Q$.
    Because $P'[s'_i, s'_j]$ is $6$-good, we have by \cref{lem:equal_label_curves} that the label curves of $\Sigma_\delta(P'[s'_i, s'_j])$ and $\Sigma_{2\delta}(Q[y, y'])$ are the same for any subcurve $Q[y, y']$ within \f distance $\delta$ of $P'[s'_i, s'_j]$
    Moreover, if the label curves of $\Sigma_\delta(P'[s'_i, s'_j])$ and $\Sigma_{2\delta}(Q[y, y'])$ are equal for some subcurve $Q[y, y']$, then $d_F(P'[s'_i, s'_j], Q[y, y']) \leq (\alpha+3) \delta$.

    Since $Q[y_1, y_2]$ lies within distance $6\delta$ of $P'(s'_i)$ and hence lies inside a single grid cell, we have that the label curves of all subcurves $Q[y, y']$ starting at a point $Q(y)$ on $Q[y_1, y_2]$ are equal.
    Therefore, taking the maximal subcurve $Q[y'_1, y'_2]$ such that the label curves of all $\Sigma_{2\delta}(Q[y, y'])$ with $y \in [y_1, y_2]$ and $y' \in [y'_1, y'_2]$ are equal to that of $\Sigma_\delta(P'[s'_i, s'_j])$, results in an $(\alpha+3, \delta)$-exit set for all points in $\{0\} \times [y_1, y_2]$.
        
    We present a data structure which can construct such an exit set in constant time:
    
    \begin{lemma}
        We can preprocess $P$ and $Q$ in $O(n)$ time, such that given a subcurve $P'$ of $P$ and a value $y \in [0, 1]$, we can report the interval $[y_1, y_2]$ such that $\Sigma_\delta(P')$ and $\Sigma_{2\delta}(Q[y, y'])$ have the same label curves for any $y' \in [y_1, y_2]$ in constant time.
    \end{lemma}
    \begin{proof}
        We make use of the string matching data structure of Chan and Rahmati~\cite{chan18improved_approximation}.
        They show that after preprocessing two strings in linear time, any two substrings can be checked for equality in constant time.
        The strings we preprocess are those defined by the extrema of the label curves of $\Sigma_\delta(P)$ and $\Sigma_{2\delta}(Q)$.
        Constructing these signatures (and their label curves) takes $O(n)$ time~\cite{driemel15clustering}.
    
        Let $A$ and $B$ be the strings defined by the label curves of $\Sigma_\delta(P)$ and $\Sigma_{2\delta}(Q)$, respectively.
        Let $A[x_1, x_2]$ be the substring of $A$ defined by extrema of the label curve of $P$ corresponding to $\delta$-signature vertices on $P[x_1, x_2]$.
        Define $B[y_1, y_2]$ analogously.
        For a given query subcurve $P[x_1, x_2]$, the extrema of the label curve of $\Sigma_\delta(P')$ define the string $\ell(P(x_1)) \circ A[x_1, x_2] \circ \ell(P(x_2))$, with consecutive duplicate characters removed.
        This affects only the first and last characters in $A[x_1, x_2]$.
        Symmetrically, given a value $y$, a subcurve $Q[y, y']$ defines the string $\ell(Q(y)) \circ B[y, y'] \circ \ell(Q(y'))$, with consecutive duplicate characters.
        This also affects only the first and last characters in $B[y_1, y_2]$.

        We perform three tests, with which we determine the interval $[y_1, y_2]$ containing all $y'$ for which the label curves of $\Sigma_\delta(P)$ and $\Sigma_{2\delta}(Q[y, y'])$ are equal.
        We first test whether $\ell(P(x_1)) = \ell(Q(y))$ by comparing the labels in constant time.
        We define $A'$ to be the maximal substring of $A[x_1, x_2]$ that does not start at $\ell(P(x_1))$ and does not end at $\ell(P(x_2))$.
        Because the length of $A'$ is fixed, there is only one possible substring $B'$ of $B$ that can match to it.
        Specifically, this is the substring of length $|A'|$ starting at the character after $\ell(Q(y))$.
        We test if $A' = B'$ in constant time.
        Lastly, we test if $\ell(P(x_2)) = \ell(Q(y'))$.
        If any of these tests fail, then there is no $y'$ for which $\Sigma_\delta(P)$ and $\Sigma_{2\delta}(Q[y, y'])$ have the same label curves.
        Otherwise, the maximum interval $[y_1, y_2]$ containing $y'$ for which $Q[y_1, y_2]$ lies inside cell $\ell(P(x_2))$ meets the query requirements.
        This range can be obtained from the parameterization of $B$ in constant time.
    \end{proof}

    The above gives an $(\alpha+3, \delta)$-exit set $\{1\} \times [y'_1, y'_2]$ for all points in $\{0\} \times [y_1, y_2]$, with respect to $P'[s'_i, s'_j]$ and $Q$.
    This exit set is also an $(\alpha+3, \delta)$-exit set for $\{z\}$, with respect to $P'[0, s'_j]$ and~$Q$.

    Finally, we extend the exit set to be an exit set for $\{z\}$ with respect to all of $P'$ and $Q$.
    The diameter of $Q[y'_1, y'_2]$ is at most $\alpha \delta$, since it lies inside a single grid cell.
    Similar to the first exit set, note that because $P'[s'_j, 1] = P[s_j, s_{j+1}]$ is the subcurve between two consecutive signature vertices, we have $d_F(P'[s'_j, 1], \overline{P'(s'_j) P'(1)}) \leq \delta$ by \cref{lem:distance_to_signature}.
    We construct a $(\alpha / 2 + 3, 2\delta)$-exit set for $\{0\} \times [y'_1, y'_2]$ with respect to $\overline{P'(s'_j) P'(1)}$ and $Q$, with the data structure of \cref{lem:ds_exit_set_segment}.
    The returned exit set is also an $(\alpha + 7, \delta)$-exit set with respect to $P'[s'_j, 1]$ and $Q$.
    Furthermore, by construction, this exit set is an $(\alpha + 7, \delta)$-exit set for $\{z\}$ with respect to $P'$ and $Q$.
    \Cref{lem:constructing_exit_interval} follows.
    
    \begin{lemma}
    \label{lem:constructing_exit_interval}
        We can preprocess $P$ and $Q$ in $O(n \log n)$ time, such that given an interior $6$-good subcurve $P'$ of $P$ and a point $z = (0, y)$, we can report a $(\alpha + 7, \delta)$-exit set for $\{z\}$ with respect to $P'$ and $Q$ in $O(\log n)$ time.
        The exit set consists of a single connected component.
    \end{lemma}

\subsection{Constructing exit sets in the general case}
\label{sub:exit_set_general_case}
    
    We extend the result of \cref{lem:constructing_exit_interval} to the general case, where we wish to construct an $(\alpha, \delta)$-exit set $E_\alpha(S)$ for a given set $S \subseteq \{0\} \times [0, 1]$ with respect to $P$ and $Q$, making no assumptions on $P$.
    We do make the assumption that $S$ consists of only $O(n)$ connected components.
    This assumption holds for example for $\F(P, Q) \cap (\{0\} \times [0, 1])$, and thus is a natural assumption to make.
    The algorithm takes $O((n^2 / \alpha^2) \log n)$ time, roughly matching the current state-of-the-art for the discrete \f distance~\cite{chan18improved_approximation} (although our algorithm is only for the one-dimensional case).
    We assume $\alpha \geq 8$, to make expressions involving $\alpha$ hold.
    For $\alpha \in [1, 8)$ we can traverse the entire free space in $O(n^2)$ time, which is faster than the above bound for these values of $\alpha$.

    We first bound the number of $6$-bad $\delta$-signature vertices of $P$ and $7$-bad vertices of $Q$.
    The total number of $c$-bad points on $P$ and $Q$ can of course be $\Theta(n)$.
    However, Chan and Rahmati~\cite{chan18improved_approximation} show that by shifting the curves it is always possible to reduce the number of $c$-bad points to at most $cn / \alpha$, and show that such a shift can be computed in $O(n)$ time.
    We assume that the total number of $6$-bad $\delta$-signature vertices of $P$, and $7$-bad vertices of $Q$, is $O(n / \alpha)$.
    This is achieved by computing a shift with at most $7 \cdot n / \alpha = O(n / \alpha)$ $7$-bad vertices on both curve combined.
    
    With the fact that every $\delta$-signature vertex of $P$ matches close to vertices of $Q$ (see \cref{lem:signature_matchings}), we have that the $6$-bad $\delta$-signature vertices of $P$ match to points on edges of $Q$ with a $7$-bad endpoint.
    Hence these $6$-bad signature vertices have only $O(n / \alpha)$ possible edges to match to.
    We use these vertices as ``bottlenecks'' in the free space, and we compute exit sets between bottlenecks.
    The subcurves between these bottlenecks are interior $6$-good, allowing us to use the data structure of \cref{lem:constructing_exit_interval} for efficiently constructing exit sets.

    We identify the $6$-bad $\delta$-signature vertices $\sigma_{i_1}, \dots, \sigma_{i_k}$ of $P$, and let $s_{i_1}, \dots, s_{i_k}$ be the values determining these vertices.
    The possible edges a signature vertex can match to correspond to certain \emph{candidate passages} in free space.
    We report these passages in $O(\log n + n / \alpha)$ time, after $O(n \log n)$ time preprocessing, by storing the vertices of $Q$ in a balanced binary search tree ordered by value, and performing a range reporting query to report the vertices, and their incident edges, that are close to the signature vertex.
    
    We iteratively go through the $6$-bad $\delta$-signature vertices of $P$.
    For each vertex $\sigma_{i_j}$, we construct an $(\alpha, \delta)$-exit set $E_j = E_\alpha(S_j)$ with respect to $P[s_{i_j}, s_{i_{j+1}}]$ and $Q$, given the set $S_j$ that is the intersection between $E_{j-1}$ and the candidate passages of $\sigma_{i_j}$.
    We consider both $P(0)$ and $P(1)$ to be $6$-bad, even if they are not, and let $\sigma_{i_0} = P(0)$ and $\sigma_{i_{k+1}} = P(1)$.
    We also set $S_0 = S$.
    
    Suppose we have constructed a set $S_j$ and need to construct a set $E_j$.
    For each connected component $\{s_{i_j}\} \times [y_1, y_2]$ in $S_j$, we construct an $(\alpha, \delta)$-exit set $E_\alpha(\{(0, y_1)\})$ with respect to $P[s_{i_j}, s_{i_{j+1}}]$ and $Q$.
    Because $P[s_{i_j}, s_{i_{j+1}}]$ is interior $6$-good, this takes only $O(\log n)$ time with the algorithm of \cref{lem:constructing_exit_interval} (setting $\alpha \gets \alpha-7$).
    Furthermore, $E_\alpha(\{(0, y_1)\})$ consists of a single connected component.
    We let $E_j = E_\alpha(S_j)$ be the union of all constructed exit sets, which we can construct in $O(|S_{j-1}| \log n)$ time by sorting the endpoints and scanning.
    This is $O(n \log n)$ for $j=1$ and $O((n / \alpha) \log n)$ for all other~$j$.
    
    If $j = k$, then $E_j$ is an $(\alpha, \delta)$-exit set for $S_0 = S$ (with respect to $P$ and $Q$) and we are done.
    Otherwise, we turn $E_j$ into a suitable set $S_{j+1}$ and repeat the process for $E_{j+1}$.
    To turn $E_j$ into $S_{j+1}$, we intersect $E_j$ with the union of the $O(n / \alpha)$ candidate passages of $\sigma_{j+1}$.
    We report these passages in $O(\log n + n / \alpha)$ time, after which we construct their union in $O((n / \alpha) \log n)$ time by sorting the endpoints of the passages.
    Using a single scan, the intersection of the two sets can be computed in $O(n / \alpha)$ time, to obtain $S_{j+1}$ in $O((n / \alpha) \log n)$ time in total.
    
    The above construction is repeated $k = O(n / \alpha)$ times, so \cref{thm:computing_exit_sets} follows.

    \computingExitSets*

\section{Concluding remarks}

    We presented faster approximation algorithms for computing the continuous \f distance between curves.
    In particular, we gave the first strongly-subquadratic $n^\eps$-approximation algorithm, for any constant $\eps \in (0, 1/2]$.
    For the one-dimensional case, we gave a further improvement over the general algorithm.
    Our curve simplification procedure proves to be a valuable tool in speeding up the one-dimensional algorithm, and we are confident that future approximation algorithms can make use of the simplification as well.
    
    Although we used the simplification for curves in general dimensions, lowering the complexity of the reachable free space by a factor $\alpha$ (in terms of blocks), we have not been able to take advantage of this lower complexity yet aside from a naive traversal.
    We expect that our one-dimensional algorithm can be adapted to work in higher dimensions as well, taking advantage of the lower complexity to yield a faster algorithm.
    The main hurdle is the time required to compare two ``good'' subcurves.
    In one dimension this takes merely constant time using string matching.
    In higher dimensions however, where we define a curve as good if the endpoints of all its monotone pieces are good, we currently see no sublinear time algorithm for this, even after preprocessing.
    The problem here is that a monotone piece may increase the complexity of the label curve beyond linear, and obtaining a linear bound even when approximations are allowed seems challenging.

\bibliography{bibliography}

\end{document}